\documentclass[12pt]{caltech_thesis}
\usepackage[hyphens]{url}
\usepackage{lipsum}
\usepackage{graphicx}
\usepackage{amsthm}
\usepackage{bbm}
\usepackage{thmtools}
\usepackage{thm-restate}

\usepackage{todonotes}

\usepackage[utf8]{inputenc}
\usepackage[T1]{fontenc}
\usepackage{newtxtext,newtxmath}
\DeclareUnicodeCharacter{2009}{\,}

\usepackage[numbers,sort&compress]{natbib}
\usepackage[sectionbib]{bibunits}
\defaultbibliographystyle{plainnat}
\usepackage{enumitem}

\usepackage{algorithm}
\usepackage{algorithmic}

\newtheorem{theorem}{Theorem}[section]
\newtheorem{lemma}{Lemma}[section]
\newtheorem{corollary}{Corollary}[section]
\newtheorem{proposition}{Proposition}[section]
\newtheorem{definition}{Definition}[section]
\newtheorem{remark}{Remark}[section]
\newtheorem{problem}{Problem}
\newtheorem{result}{Result}
\newtheorem{fact}{Fact}[section]

\usepackage{hyperref}
\hypersetup{
    colorlinks=true,
    linkcolor=blue,
    filecolor=magenta,      
    urlcolor=cyan,
    citecolor=blue,
}

\begin{document}

\title{Efficient and SPAM-Robust Ansatz-Free Lindbladian Learning}
\author{Savar Dayal Sinha}

\degreeaward{Bachelor of Science}                 
\university{California Institute of Technology}    
\address{Pasadena, California}                     
\unilogo{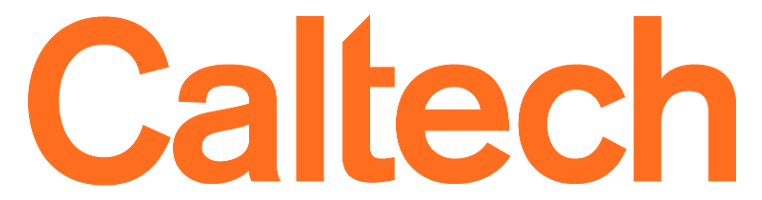}                                 
\copyyear{2026}  
\defenddate{May 29, 2026}          

\orcid{0000-0003-0155-0883}

\rightsstatement{All rights reserved}

\maketitle[logo]

\begin{acknowledgements}
    First and foremost, I must express my deepest gratitude to Professor Yu Tong and Professor John Preskill. Over the past two years, Professor Tong's guidance opened the door to the world of quantum learning theory; I would never have found my way into these exciting research problems without his encouragement and insights. I am equally grateful to John Preskill for the incredible opportunity to join his group at Caltech and for his continued mentorship throughout this process. Together with Professor Nat Tantivasadakarn, these mentors have provided invaluable tutelage during my undergraduate career.

    I would also like to thank the members of the IQIM community for providing such a vibrant and intellectually stimulating environment; the discussions and support from the group were invaluable throughout the development of this thesis.
    
    I am grateful to the faculty who first introduced me to the foundations of quantum computing through their coursework, including Alexei Kitaev, John Preskill, and Hsin-Yuan (Robert) Huang.
    
    On a personal note, I want to thank Ritvik Teegavarapu and Kieran Vlahakis for their friendship and for making this undergraduate experience truly memorable. Finally, I owe everything to my parents, Reena and Kislay Sinha, and my sister, Drishti, for their steadfast support.
\end{acknowledgements}

\begin{abstract}
    Describing the dynamics of open systems is essential for fault-tolerant quantum computation. Under Markovian assumptions, we can characterize dissipative dynamics via the Lindbladian. Using Bell sampling, we provide an efficient, ansatz-free Lindbladian learning algorithm with polynomial-time classical postprocessing. Motivated by the prevalence of state preparation and measurement (SPAM) noise on near-term devices, we also introduce the first efficient SPAM-robust protocol capable of learning the gauge-independent components of sparse Lindbladians to arbitrary precision in the presence of constant-order SPAM error. In doing so, we provide the first rigorous characterization of the gauge degrees of freedom in noisy Lindbladian learning, precisely identifying which components remain learnable under SPAM noise.
\end{abstract}








\tableofcontents

\mainmatter

\chapter{Introduction}




Practical-scale quantum computing requires large numbers of qubits. However, as quantum devices continue to scale to hundreds of qubits, they have become increasingly susceptible to noise, making quantum error correction essential. While substantial progress has been made in this area, achieving fault tolerance requires careful characterization of the underlying generators governing the dynamics of quantum hardware in order to engineer devices to mitigate the effects of this noise.

In the idealized setting of noiseless Hamiltonian evolution, this task has been extensively studied under the Hamiltonian learning problem \cite{holzapfel_scalable_2015, arunachalam_testing_2025, ma_learning_2024, hu_ansatz-free_2025, sinha_improved_2025, zubida_optimal_2021, caro_learning_2024, huang_learning_2023, dutkiewicz_advantage_2024, bakshi_structure_2024, qi_determining_2019, haah_optimal_2022, chen_learning_2025, anshu_sample-efficient_2021, evans_scalable_2019, li_hamiltonian_2020, haah_learning_2024, gu_practical_2024}. Existing Hamiltonian learning approaches fall into one of two categories: learning from unitary dynamics \cite{arunachalam_testing_2025, ma_learning_2024, hu_ansatz-free_2025, sinha_improved_2025, zubida_optimal_2021, caro_learning_2024, huang_learning_2023, dutkiewicz_advantage_2024, bakshi_structure_2024, haah_learning_2024, gu_practical_2024} or learning from specialized states (e.g. Gibbs states or steady states) based on the Hamiltonian \cite{qi_determining_2019, haah_optimal_2022, chen_learning_2025, anshu_sample-efficient_2021, evans_scalable_2019, li_hamiltonian_2020, haah_learning_2024, gu_practical_2024}. Under various sets of assumptions, several Hamiltonian learning algorithms have even achieved optimal scaling (e.g. shortest evolution time for learning from dynamics) \cite{huang_learning_2023, ma_learning_2024, hu_ansatz-free_2025, haah_optimal_2022, bakshi_structure_2024}. Nevertheless, since these algorithms operate under the assumption of noiseless evolution, we cannot use them to capture the dynamics of quantum devices in general, which are typically dissipative in nature.

This limitation motivates the study of Lindbladian learning and certification as a framework for characterizing open systems \cite{liu_robust_2025, onorati_fitting_2023, bairey_learning_2020, stilck_franca_efficient_2024, ivashkov_ansatz-free_2026, cai_optimal_2026}. A range of approaches for learning Lindbladians has been proposed, including convex optimization \cite{onorati_fitting_2023, liu_robust_2025}, steady-state methods \cite{bairey_learning_2020}, and derivative estimation \cite{stilck_franca_efficient_2024, ivashkov_ansatz-free_2026}. However, the exponential state space renders many of these algorithms intractable for systems with large numbers of qubits, requiring structural assumptions like locality to operate in polynomial evolution time \cite{stilck_franca_efficient_2024}. Recent work by Ivashkov et al. addresses the nonlocal setting with an ansatz-free algorithm achieving optimal sample complexity \cite{ivashkov_ansatz-free_2026}. Nevertheless, their derivative-based approach requires solving a large linear system, leading to classical post-processing costs that can scale exponentially in the number of qubits in the worst case. Naturally, this motivates the first question we aim to answer in this thesis: can we learn Lindbladians with both efficient evolution time and classical post-processing without any prior structural assumptions?

A further practical challenge arises from state preparation and measurement (SPAM) errors, which significantly impact near-term quantum devices. While prior work has incorporated robustness to SPAM noise, these results typically rely on restrictive assumptions, such as small or local error models \cite{stilck_franca_efficient_2024, ivashkov_ansatz-free_2026}. More general techniques, such as gate set tomography \cite{liu_robust_2025}, can handle arbitrary SPAM error but incur exponential time overhead, making them impractical for large systems.

Beyond these computational challenges, SPAM error introduces a more fundamental issue: it induces gauge degrees of freedom in the learning process, under which distinct Lindbladians become operationally indistinguishable \cite{liu_robust_2025}. As a result, certain components of the Lindbladian are inherently unlearnable in the presence of SPAM noise, regardless of the learning protocol. This raises a second central question: what aspects of a Lindbladian are identifiable under general SPAM error, and can the learnable components be extracted efficiently?

In this thesis, we address both questions. First, we develop an algorithm for learning sparse, potentially nonlocal Lindbladians with efficient evolution time and polynomial classical post-processing. Second, we provide a complete characterization of the gauge degrees of freedom induced by SPAM error, thereby identifying precisely which components of a Lindbladian are learnable and which are fundamentally unobservable. Building on this characterization, We present an efficient protocol for recovering all learnable parameters in the presence of unstructured SPAM noise.

In particular, we will begin by laying out the foundational concepts of open system evolution and the properties of quantum channels in Chapter 2. In Chapter 3, we will analyze Lindbladian learning in the noiseless setting to demonstrate an efficient Lindbladian learning algorithm with polynomial classical post-processing time. In Chapter 4, we will demonstrate the effect of SPAM error on the Lindbladian learning process, first characterizing the coefficients of the Lindbladian that are unlearnable in the presence of SPAM error before outlining an algorithm to learn the remaining components of the Lindbladian efficiently. Finally, in Chapter 5, we will discuss the overall conclusions of this thesis as well as the future directions and applications of these findings.

\section{Notation and Conventions}

Before introducing the theory, we will first establish the notation and conventions that will be used throughout the thesis.

\subsection{Pauli Group}

We define the $n$-qubit Pauli group (without phase) as $\mathcal P_n = \{I, X, Y, Z\}^{\otimes n}$, where
$$I = \begin{pmatrix}
    1 & 0 \\
    0 & 1
\end{pmatrix}, \quad X = \begin{pmatrix}
    0 & 1 \\
    1 & 0
\end{pmatrix}, \quad  Y = \begin{pmatrix}
    0 & -i \\
    i & 0
\end{pmatrix}, \quad  Z = \begin{pmatrix}
    1 & 0 \\
    0 & -1
\end{pmatrix}$$
We will enumerate Paulis $P_b \in \mathcal P_n$ via a $2n$-bit string $b = (b^x, b^z)$ where $b^x, b^z \in \mathbb F^n_2$, where
$$P_b = i^{b^x \cdot b^z} \bigotimes_{i = 1} ^n X^{b^x_i}Z^{b^z_i}$$
where $b^x \cdot b^z = \sum_{i = 1} ^n b^x_ib^z_i \pmod 2$ is the bitwise inner product modulo 2. For two Paulis $P_a, P_b$, we encode their commutation relations via the symplectic inner product $\langle \cdot, \cdot \rangle$ such that
$$P_{a}P_b = (-1)^{\langle a, b\rangle}P_bP_a,$$
where for two $2n$-bit strings $a = (a^x, a^z)$ and $b = (b^x, b^z)$, we define
$$\langle a, b \rangle = a^x \cdot b^z - a^z \cdot b^x \pmod 2$$

\subsection{Bell States}

For our Bell sampling procedures, we define the four two-qubit Bell states as follows:
\begin{align*}
    |\Phi^+\rangle &= \frac{|00\rangle + |11\rangle}{\sqrt 2} & |\Phi^-\rangle &= \frac{|00\rangle - |11\rangle}{\sqrt 2} \\
    |\Psi^+\rangle &= \frac{|01\rangle + |10\rangle}{\sqrt 2} & |\Psi^-\rangle &= \frac{|01\rangle - |10\rangle}{\sqrt 2}
\end{align*}
As a shorthand, we will typically write $|\Phi_n\rangle = |\Phi^+\rangle^{\otimes n}$.

\subsection{Norms}

We will utilize various norms to characterize superoperators. For a given superoperator $\Phi: \mathbb C^{n \times n} \to \mathbb C^{m \times m}$, we define its diamond norm $\|\cdot\|_\diamond$ as follows:
$$\|\Phi\|_\diamond = \max_{X : \|X\|_1 \le 1}\|(\Phi \otimes \mathcal I)X\|_1,$$
where
$$\|A\|_1 = \text{Tr}\left(\sqrt{A^\dag A}\right)$$
is the trace norm, $X \in \mathbb C^{n^2 \times n^2}$, and $\mathcal I: \mathbb C^{n \times n} \to \mathbb C^{n \times n}$ denotes the identity superoperator.









\chapter{Open Systems and Lindbladians}

In quantum mechanics, closed systems are evolved in unitary fashion according to the operation $U(t) = e^{-iHt}$, where $H$ is the Hamiltonian of the system. However, in practice, physical systems we deal with are rarely closed, motivating the need for a theory of open-systems. In this chapter, we will establish the foundations of open-systems by first introducing the notion of quantum channels before discussing Markovian time evolution under the Lindbladian.

\section{Quantum Channel Formalism}

While closed systems are transformed via unitary operations, open systems are evolved via quantum channels. Consequently, while the purity of states is preserved as a closed system evolves, the same does not hold true in general for open systems. For instance, a spin can interact with oncoming photons, performing measurements on its state that cause any potential superposition to decohere into a mixed state. Consequently, to represent transformations for open quantum systems, we require a more robust formalism known as quantum channels. We will discuss four formalisms for investigating such channels: CPTP maps, isometries, Kraus operator decompositions, and Choi matrices.

\subsection{CPTP Maps}

To precisely define what a quantum channel is, we look at its abstract mathematical definition.
\begin{definition}[Quantum Channel]
    A superoperator $\Phi: \mathcal B(\mathcal H_L) \to \mathcal B(\mathcal H_L)$ is a quantum channel if and only if it is both completely positive and trace-preserving (CPTP).
\end{definition}
As the name implies, trace-preserving simply implies that for any operator $A$, the superoperator $\Phi$ satisfies
$$\text{Tr}(\Phi[A]) = \text{Tr}(A)$$
A map $\Phi: \mathcal B(H_L) \to \mathcal B(H_L)$ is completely positive if and only if for any operator $A \in \mathcal B(\mathcal H_L \otimes \mathcal H_R)$, we have that
$$A \succeq 0 \implies (\Phi \otimes \mathcal I_R)(A) \succeq 0,$$
where $\mathcal I_R$ is the identity map on the Hilbert space $\mathcal H_R$. Intuitively, this means that if we extend our density matrix to include an auxiliary system unaffected by the quantum channel, this transformed system must remain a valid density matrix. While this definition is rigorous, it does little to explain intuitively what a quantum channel represents. To understand this, we require the following formalisms.

\subsection{Isometry}

An open system only exists as such in reference to a larger environment, which collectively form a closed system. Consequently, we can view a quantum channel as merely the result of a unitary transformation when we trace out this auxiliary system. Formally, for any quantum channel $\mathcal E$, we can define an isometry $V: \mathcal H_A \to \mathcal H_A \otimes \mathcal H_E$ such that
$$\mathcal E[\rho] = \text{Tr}_E(V \rho V^\dag)$$
This isometry $V$ is known as the Stinespring dilation.

\subsection{Kraus Operators}

To conveniently represent a superoperator, we can use the Kraus operator decomposition. Defining an orthonormal basis $\{|\xi_i\rangle_E\}$ that spans $\mathcal H_E$, we have that
$$\mathcal E[\rho] = \text{Tr}_E(V \rho V^\dag) = \sum_i (I_A \otimes \langle\xi_i |_{E})V\rho V^\dag (I_A \otimes |\xi_i\rangle_E) = \sum_i K_i \rho K_i^\dag$$
Here, these $K_i$ operators are known as Kraus operators. Given a Kraus operator decomposition, we can also construct the Stinespring dilation by defining $V$ such that
$$V|\psi \rangle = \sum_i K_i|\psi\rangle \otimes |i\rangle$$
While any completely positive superoperator admits a Kraus decomposition, we require additional constraints on $K_i$ for the quantum channel to be trace-preserving. In particular, since $\mathcal E$ maps density matrices to other density matrices, we must have that $\text{Tr}(\mathcal E[\rho]) = 1$. By cyclicity of trace, we have that this gives us the following sufficient and necessary condition for a Kraus operator decomposition to correspond to a quantum channel:
$$\sum_i K_i^\dag K_i = I$$
It should be noted that the similar statement $\sum_i K_iK_i^\dag = I$ does not hold in general, instead only holding when $\mathcal E[I] = I$. Channels which obey this relation, such as the depolarizing and dephasing channels, are called unital while those that do not obey this, like the amplitude-damping channel, are called non-unital. 

\subsection{Choi States}

It is often convenient to unfold superoperators into operator form, allowing us to write them as matrices. One common formalism for doing so is the Choi-Jamiołkowski isomorphism. In particular, we define the Choi state by considering the action of the quantum channel acting on one qubit of $n$ Bell-pairs, which gives us the following:
$$\mathcal J(\mathcal E) = (\mathcal I \otimes \mathcal E)(|\Phi_n\rangle\langle\Phi_n|) = \sum_{i, j} |i\rangle\langle j| \otimes \mathcal E(|i\rangle\langle j|)$$
Intuitively, we can write the Choi matrix in block form as follows:
$$\mathcal J(\mathcal E) = \begin{pmatrix}
    \mathcal E(|0\rangle\langle 0|) & \mathcal E(|0\rangle\langle 1|) & \cdots & \mathcal E(|0\rangle\langle 2^{n}-1|) \\
    \mathcal E(|1\rangle\langle 0|) & \ddots & & \vdots \\
    \vdots & & \ddots & \vdots \\
    \mathcal E(|2^n - 1\rangle\langle 0|) & \cdots & \cdots & \mathcal E(|2^n - 1\rangle\langle 2^n - 1|)
\end{pmatrix}$$
This form is useful for checking certain properties like complete-positivity that would be difficult to verify otherwise. For instance, we have that $\mathcal E$ is completely positive if and only if $\mathcal J(\mathcal E) \succeq 0$.

\section{Other Superoperator Matrix Representations}

Beyond Choi states, there exist a number of other useful formalisms for expressing superoperators in matrix form, which we will employ extensively throughout this thesis.

\subsection{Chi Matrix}

The first formalism we consider is the Chi matrix (also known as the process matrix). This matrix encodes the coefficients of a superoperator in an operator expansion of our choice. Formally, we define the Chi matrix as follows:
\begin{definition}[Chi Matrix / Process Matrix]\label{def:chi}
    For a given operator basis $\{F_i\}$, the Chi matrix for 
    superoperator $\mathcal E(A) = \sum_{i, j}\chi_{ij}F_iAF_j^\dag$ is defined as the coefficient matrix $\chi(\mathcal E) = (\chi_{ij})$
\end{definition}
For our purposes, we will always assume our operator basis to be the Pauli basis $\mathcal P_n$. We will use this matrix to provide a compact representation of the Lindbladian with respect to its expansion in the Pauli basis. Furthermore, by fixing the basis of our operator decomposition, we can formally define the notions of support and sparsity:
\begin{definition}[Support]
    For a superoperator $\mathcal E = \sum_{i, j} \chi_{ij} P_i \rho P_j$, we define the support $\mathcal S$ of $\mathcal E$ as
    $$\mathcal S = \{i: \exists j \text{ s.t. } \chi_{ij} \ne 0 \text{ or } \chi_{ji} \ne 0\}$$
\end{definition}
\begin{definition}[$M$-sparse]
    A superoperator $\mathcal E: \mathbb C^d \to \mathbb C^d$ is said to be $M$-sparse if and only if $|\mathcal S| \le M$, where $\mathcal S$ is the support of $\mathcal E$.
\end{definition}
For $n$-qubit Lindbladian learning to be efficient, we must assume $M = \text{poly}(n)$ in our operator basis of choice (i.e. the Pauli basis) for learning to be efficient, as we would have to learn an exponential number of coefficients otherwise.

\subsection{Pauli Transfer Matrix}

The second matrix formalism we are interested in is the Pauli Transfer Matrix (PTM). The PTM describes the action of a given supeoperator on each Pauli. This formalism is especially convenient for Pauli-diagonal superoperators like Pauli channels or Pauli Lindbladians, as the resulting PTM is diagonal. We formally define the PTM as follows:
\begin{definition}[Pauli Transfer Matrix]\label{def:ptm}
    For a given superoperator $\mathcal E(P_j) = \sum_{i}T_{ij} P_i$, where $P_i, P_j \in \mathcal P_n$, we define the Pauli transfer matrix as $\mathcal R(\mathcal E) = (T_{ij})$, where
    $$T_{ij} = \frac{1}{4^n}\text{Tr}(P_i \mathcal E(P_j))$$
\end{definition}
Furthermore, the PTM is amenable to algebraic manipulation. Since the Pauli transfer matrix acts as a linear transformation on the Pauli basis, we can write the PTM for the composition $\mathcal E_1 \circ \mathcal E_2$ as the product of their individual Pauli transfer matrices:
$$\mathcal R(\mathcal E_1 \circ \mathcal E_2) = \mathcal R(\mathcal E_1)\mathcal R(\mathcal E_2)$$

\section{Lindbladians}

Just as the Hamiltonian $H$ generates the evolution of closed systems, the Lindbladian $\mathcal L$ analogously describes the evolution of open systems. In particular, this evolution is governed by the Gorini-Kossakowski-Sudarshan-Lindblad (GKSL) equation, expressed compactly as follows:
$$\dot \rho = \mathcal L[\rho]$$
Solving this differential equation gives us that
$$\rho(t) = e^{t\mathcal L}[\rho(0)]$$

\subsection{Derivation}

To derive the form of the Lindblad operator, we consider the evolution of the state under small time-steps \cite{preskill_physics_nodate}. Assuming Markovian evolution, we have that the state $\rho$ at time $t + \delta t$ is entirely determined by the state at $t$. In practice, this assumption may not hold, as it is possible for information to flow out of a system at some time and return at a later time. Nevertheless, the Markovian description provides a good approximation in many cases. Under this assumption, we have that there exists a quantum channel $\mathcal E = e^{\delta t\mathcal L}$ that takes
$$\mathcal E: \rho(t) \mapsto \rho(t + \delta t) \approx \rho(t) + \delta t \mathcal L[\rho(t)] + \mathcal O(\delta t^2)$$
Expanding out the quantum channel in terms of its Kraus operators, we have that
$$\mathcal E[\rho(t)] = \sum_{j} K_j \rho(t) K_j^\dag \approx \rho(t) + \delta t \mathcal L[\rho(t)]$$
Without loss of generality, we can fix $K_0 = I + \mathcal O(\delta t)$, meaning that for each of the remaining Kraus operators with $j > 0$, $K_j = \mathcal O(\sqrt{\delta t})$. Intuitively, we can interpret this as meaning that the probability of a discontinuous jump to some other state occurring over the time $dt$ is given as $\mathcal O(dt)$. With these assumptions in mind, we have the following ansätze for our Kraus operators:
\begin{align*}
    K_j &= \begin{cases}
        I + \delta t (-iH + K) & j = 0 \\
        \sqrt{\delta t} L_j & j > 0
    \end{cases}
\end{align*}
where $H, K$ are Hermitian matrices and $L_j$ are jump operators. These $H, K, L$ operations are all taken be constant order in $\delta t$. Substituting this into the Kraus decomposition of $\mathcal E$, we have that
$$\mathcal E[\rho(t)] = \rho(t) + \delta t\left[-i[H, \rho(t)] + \sum_{j > 0} L_j \rho(t) L_j^\dag + \{K, \rho(t)\}\right] + \mathcal O(\delta t^2)$$
Furthermore, by the normalization property of quantum channels ($\sum_{j \ge 0} K_j^\dag K_j = I$), we find that
$$K = -\frac{1}{2}\sum_{j > 0} L_j^\dag L_j$$
Altogether, comparing the linear order term of $\mathcal E[\rho(t)]$ to the Taylor expansion obtained previously, we obtain the following canonical form of the Lindbladian:
\begin{equation}\label{eq:lindblad_diagonal}
    \mathcal L[\rho] = -i[H, \rho] + \sum_{j} \left[L_j \rho L_j^\dag -\frac{1}{2} \{L_j^\dag L_j, \rho\}\right]
\end{equation}
Physically, $H$ represents the Hamiltonian of the open system, driving a unitary style of time evolution, while the $L_j$ terms represent instantaneous transitions to other states driven by interactions between the system and its environment, such as projective measurements or Pauli errors, with the anticommutator term serving as an appropriate normalization term. For this reason, we often decompose the Lindbladian as
$$\mathcal L = \mathcal L_H + \mathcal L_D,$$
where $\mathcal L_H = -i[H, \rho]$ is the Hamiltonian term and $\mathcal L_D = \sum_{j} \left[L_j \rho L_j^\dag -\frac{1}{2} \{L_j^\dag L_j, \rho\}\right]$ is the dissipative term. The form in (\ref{eq:lindblad_diagonal}) can be thought of as writing the Lindbladian in a diagonal basis. If we wish to use any arbitrary operator basis $\{F_m\}$ for our jump operators, we can alternatively write
\begin{equation}\label{eq:lindblad_arbitrary}
    \mathcal L[\rho] = -i[H, \rho] + \sum_{a, b} \gamma_{ab} \left[F_a \rho F_b^\dag -\frac{1}{2} \{F_b^\dag F_a, \rho\}\right],
\end{equation}
where $\gamma$ forms a positive semi-definite matrix. We can then convert this form back into the other by diagonalizing this matrix and using the orthogonal eigenvectors weighted by the square root of the eigenvalues to define $L_j$.

\subsection{Conversion Between Formalisms}

For our purposes, we are interested in working with the Chi matrix rather than coefficients in (\ref{eq:lindblad_arbitrary}). The Chi matrix coefficients will therefore possess additional algebraic structure that should be analyzed more carefully. Referring to (\ref{eq:lindblad_arbitrary}), we have that choosing the operator basis $\{P_a\}$, we have that
$$\chi_{ab} = \gamma_{ab},$$
for $P_a, P_b \ne I$. This arises trivially due to the fact that the only term in (\ref{eq:lindblad_arbitrary}) with non-identity operators acting on both sides corresponds to $\gamma_{ab}P_a\rho P_b$. Similarly, for $P_a = P_b = I$, we observe that
$$\gamma_{00}\left[\rho - \frac{1}{2}\{I, \rho\}\right] = 0$$
Consequently, $\gamma_{00} = 0$ for all Lindbladians. Hence, we have that $\chi_{00}$ is dependent on other elements of the Chi matrix. In particular, since $H$ is traceless, the only contribution arises from
$$-\sum_{a = b}\frac{\gamma_{ab}}{2}\{P_b P_a, \rho\} = -\sum_j \gamma_{jj} \rho$$
Hence, we have that when $P_a = P_b = I$,
$$\chi_{ab} = -\sum_j \gamma_{jj} = -\sum_j \chi_{jj}$$
Lastly, we consider the elements where $P_a \ne I$ and $P_b = I$, and vice versa. For $P_b = I$ and $P_a = Q$,
$$P_a\rho P_b - \frac{1}{2}\{P_bP_a, \rho\} = \frac{1}{2}[Q, \rho]$$
For $P_a = I$ and $P_b = Q$ instead, we have
$$P_a\rho P_b - \frac{1}{2}\{P_bP_a, \rho\} = \frac{1}{2}[\rho, Q]$$
Consequently, adding these two terms together, we have
$$\frac{\gamma_{b0} - \gamma_{0b}}{2}[Q, \rho] = i\text{Im}[\gamma_{b0}][Q, \rho]$$
As we can see here, we obtain a Hamiltonian-like term, introducing a gauge degree of freedom in the definition of the Hamiltonian. In particular, this is equivalent to the inhomogeneous transform
\begin{align}
    L_j \to L_j' &= L_j + a_jI \\
    H \to H' &= H + \frac{1}{2i}\sum_{j} \gamma_j \left(a_j^*L_j - a_jL^\dag_j\right)
\end{align}
If we choose the global phase convention such that $\text{Tr}(L_j) \in \mathbb R, \text{Tr}(L_j) \ge 0$, we can then choose $a_i$ such that each $L_i'$ is traceless, in which case
$$\text{Tr}\left[\frac{1}{2i}\sum_{j} \gamma_j \left(a_j^*L_j - a_jL^\dag_j\right)\right] = 0,$$
since $a_j = a_j^* = \text{Tr}(L_j) = \text{Tr}(L^\dag_j)$. Consequently, under this gauge transformation, we can choose all the jump operators to be traceless whilst maintaining the traceless property of the Hamiltonian, in which case we fix all $\gamma_{0b} = \gamma_{b0} = 0$.

Consequently, all contributions to $\chi_{0b}$ and $\chi_{b0}$ must come from the Hamiltonian and anticommutator terms. Starting with the latter, we have that for anticommuting Paulis $P_a, P_b$,
$$\frac{\gamma_{ab}}{2}\{P_bP_a, \rho\} + \frac{\gamma_{ba}}{2}\{P_bP_a, \rho\} = i\text{Im}[\gamma_{ab}]\{P_bP_a, \rho\}$$
Similarly, if $P_m, P_n$ commute, we have that
$$\frac{\gamma_{ab}}{2}\{P_bP_a, \rho\} + \frac{\gamma_{ba}}{2}\{P_bP_a, \rho\} = \text{Re}[\gamma_{ab}]\{P_bP_a, \rho\}$$
If we define the Hamiltonian as $H = \sum_j h_j P_j$, we have that
$$-i[H, \rho] = -i\sum_j h_j [P_j, \rho]$$
Altogether, we therefore have the following correspondence between the two forms:
\begin{lemma}[Lindbladian Chi Matrix]
    Given a Lindbladian $\mathcal L$ in the form of (\ref{eq:lindblad_arbitrary}), we can express the Chi matrix as
    $$\chi_{ab} = \begin{cases}
        \gamma_{ab} & a, b \ne 0 \\
        -ih_b + \frac{1}{2}\sum_{m} \omega_{m, m \oplus b} \gamma_{m, m \oplus b} & a = 0, b \ne 0 \\
        ih_a + \frac{1}{2}\sum_{m} \omega_{m, m \oplus a} \gamma_{m, m \oplus a} & a \ne 0, b = 0 \\
        -\sum_{k} \gamma_{kk} & a = 0, b = 0
    \end{cases}$$
    where $\omega_{ab} \in \{\pm 1, \pm i\}$.
\end{lemma}
Furthermore, we note that the sum of the form $\frac{1}{2}\sum_{m} \omega_{m, m \oplus b} \gamma_{m, m \oplus b} \in \mathbb R$, and likewise for the sum with $a$.

\subsection{Useful Properties}

Lindbladians possess a number of useful algebraic properties that we will utilize extensively throughout this paper to demonstrate that a given superoperator constitutes a valid Lindbladian. By definition, we have that a superoperator $\mathcal L$ is a Lindbladian if and only if it generates a CPTP map $e^{t\mathcal L}$ for all $t \ge 0$.

For this to be true, there exist a number of necessary conditions that $\mathcal L$ must satisfy. To preserve the physical validity of the density matrix at all times, we have that Lindblad evolution must be hermiticity-preserving and trace-preserving. As a result of this, we have that $\text{Tr}(\mathcal L[A]) = 0$ for any operator $A$. This follows intuitively from cyclicity of trace:
$$\mathcal L(A) = -i[H, A] + \sum_i L_i A L_i^\dag - \frac{1}{2}\{L_i^\dag L_i, A\} \implies \text{Tr}[\mathcal L(A)] = 0$$
Furthermore, we have that the eigenvalues of the Lindbladian all have real parts that are negative so that all the eigenvalues of $e^{\mathcal Lt}$ have absolute value bounded above by 1.

Beyond properties concerning the validity of the Lindbladian, we can also explicitly bound the diamond norm of $\mathcal L$ under sparsity assumptions. However, to do so, we begin by first proving the following helper result:
\begin{lemma}\label{lem:superop_diamond}
    For a superoperator of the form $\Phi = A \rho B$, we can bound the diamond norm of $\Phi$ as
    $$\|\Phi\|_\diamond \le \|A\|_\infty\|B\|_\infty,$$
    where $\|\cdot\|_\infty$ denotes the operator norm.
\end{lemma}
\begin{proof}
    By the definition of the diamond norm, we have that
    $$\|\Phi\|_\diamond = \max_{X: \|X\|_1 \le 1} \|(\Phi \otimes \mathcal I)X\|_1$$
    Upper bounding the term inside the trace norm, we can apply Hölder's inequality for Schatten norms to get the following:
    $$\|(\Phi \otimes \mathcal I)X\|_1 = \|(A \otimes I)X(B \otimes I)\|_1 \le \|(A \otimes I)\|_\infty \| X \|_1 \|(B \otimes I)\|_\infty = \|A\|_\infty \| X \|_1 \|B\|_\infty$$
    Since $\|X\|_1 \le 1$, this means that
    $$\|(\Phi \otimes I)X\|_1 \le \|A\|_\infty \|B\|_\infty$$
\end{proof}
Having proven Lemma \ref{lem:superop_diamond}, we are now ready to bound the diamond norm of an $M$-sparse Lindbladian.
\begin{lemma}[Lindblad Bound]\label{lem:lindblad_diamond}
    For an $M$-sparse Lindbladian $\mathcal L = \sum_{a, b}\chi_{ab} P_a \rho P_b$ with $|\chi_{ab}| \le 1$, we bound the diamond norm of $\mathcal L$ as follows:
    $$\|\mathcal L\|_\diamond \le M^2$$
\end{lemma}

\begin{proof}
    Using triangle inequality and applying Lemma \ref{lem:superop_diamond}, we have that
    $$\|\mathcal L\|_\diamond \le \sum_{a, b}|\chi_{ab}| \|P_a \rho P_b\|_\diamond \le \sum_{a, b}|\chi_{ab}| \le M^2$$
    since $\|P\|_\infty = 1$ for any Pauli $P$.
\end{proof}







\chapter{Noiseless Lindbladian Learning}\label{chap:noiseless_lindblad_learn}

As we demonstrated in the previous chapter, the Lindbladian dictates the dynamics of open systems under Markovian assumptions. In practice, however, we typically do not know the Lindbladian of a system beforehand due to its phenomenological nature. Consequently, in this chapter, we will consider the inverse problem: reconstructing the Lindblad generator from observed dynamics. To this end, we provide an efficient algorithm for ansatz-free Lindbladian learning, wherein we first learn the Pauli support of the Lindbladian via Bell sampling before applying targeted parameter estimation techniques to recover the coefficients. We precisely formulate the Lindbladian learning problem as follows:
\begin{problem}
    Given black-box access to the time-evolution operator $e^{t\mathcal L}$ for any time $t$, can we efficiently learn the Chi matrix of an $M$-sparse Lindbladian $\mathcal L[\rho] = \sum_{i, j} \chi_{ij}P_i\rho P_j$? 
\end{problem}
For our purposes, we will assume normalization such that $|\chi_{ij}| \le 1$ for all $i, j$. Using the algorithm depicted schematically in Figure \ref{fig:lindblad_learn}, we will prove the following:
\begin{result}[Nonlocal Sparse Lindbladian Learning]
    Given an $M$-sparse Lindbladian $\mathcal L$, we can learn $\chi(\mathcal L)$ to $\epsilon$-accuracy in time
    $$T = \widetilde {\mathcal O}\left(\frac{M^4}{\epsilon^4}\right)$$
\end{result}
Note that we employ $\widetilde{\mathcal O}(\cdot)$ notation to hide poly-logarithmic factors for the sake of clarity.

\section{Learning the Support}

\begin{figure}[tbp]
    \centering
    \includegraphics[width=0.8\linewidth]{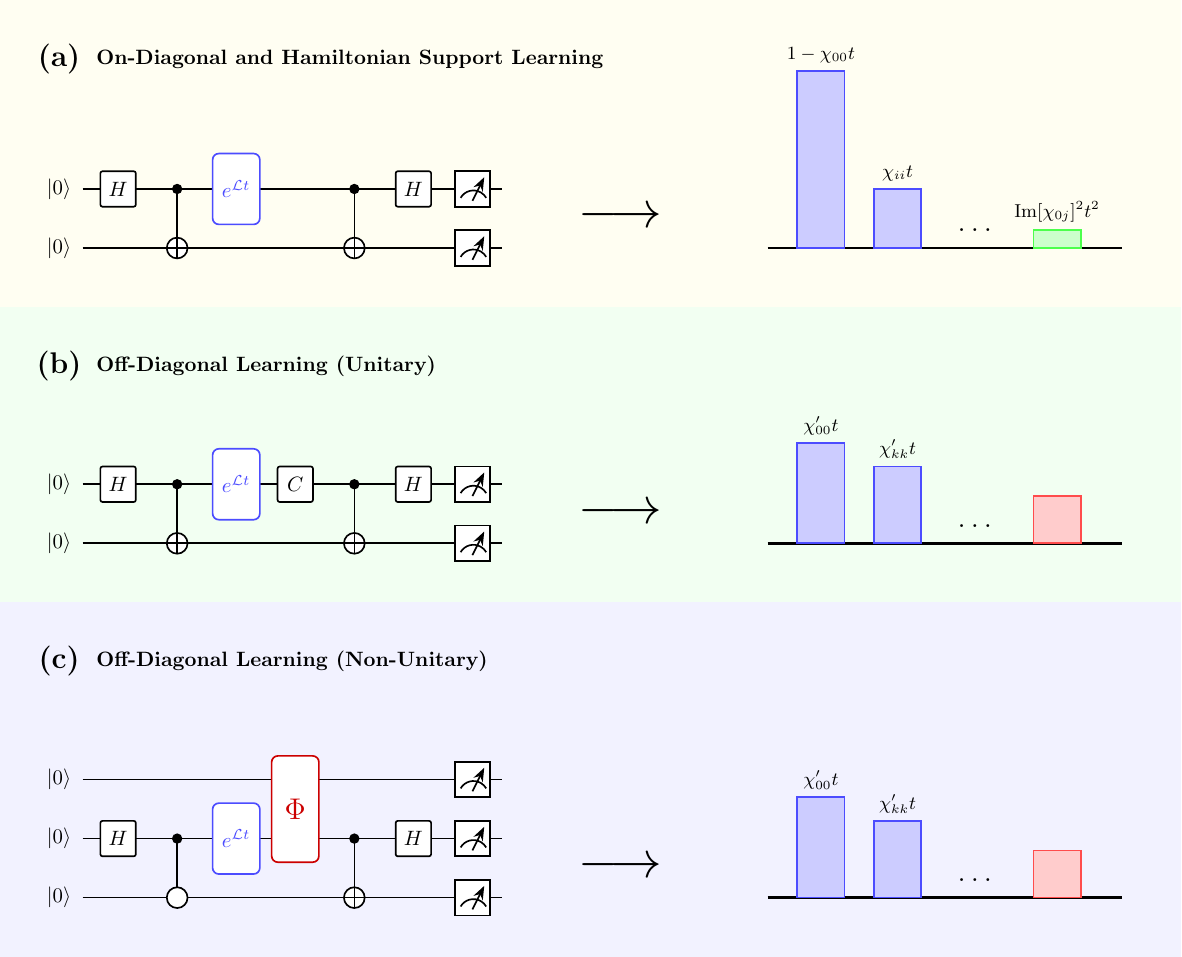}
    \caption{\textbf{Lindbladian coefficient learning via Bell sampling.} (a) Protocol for learning diagonal components of Lindbladian $\mathcal L = \sum_{i, j}\chi_{ij}P_i \rho P_j$, where $i = j$, as well as support of Hamiltonian outside dissipator support and its products. By measuring in the Bell basis after evolving a Bell state, we can sample Paulis with probability corresponding to the coefficients of the Lindbladian. (b) By performing either $C \propto P_i + P_j$ or $C \propto P_i + iP_j$, depending on which is Clifford, we can estimate the real and imaginary components of the off-diagonals, respectively, via $\chi_{00}'$ and $\chi_{kk}'$ for $k = i \oplus j$, which form linear combinations of known diagonal and unknown off-diagonal coefficients. (c) For the remaining coefficients, we perform an effective non-unitary transformation via an isometry $\Phi$ and post-selection on the measurement outcome of the extra ancilla qubit.}
    \label{fig:lindblad_learn}
\end{figure}

In Lindbladian learning, $k$-locality is commonly assumed to restrict the support of the Hamiltonian to a constant number of Pauli terms \cite{stilck_franca_efficient_2024}. However, so long as the support of the Lindbladian is guaranteed to be polynomial in size, we can always learn the Lindbladian efficiently. To demonstrate this, we utilize Bell sampling, a procedure by which we can sample from the support of a unitary operation by using $n$ Bell pairs.
\begin{fact}\label{fact:bell_sampling_unitary}
    Given a unitary $U = \sum_x \mu_xP_x$, we can sample a Pauli $P_x$ with probability $|\mu_x|^2$.
\end{fact}

\begin{proof}
    Suppose we are given a $n$ qubit unitary $U = \sum_x \mu_xP_x$. Using a system of $n$ qubits as well as $n$ ancillas, we can sample from the Pauli distribution of $U$ as follows. Initialize a state consisting of $n$ EPR pairs between the system and the ancillas: 
    $$|\Phi_n\rangle = |\Phi^+\rangle^{\otimes n}$$
    If we now apply $U \otimes I_{2^n}$, we get the following:
    \begin{align*}
        (U \otimes I_{2^n})|\Phi_n\rangle &= U \otimes I_{2^n}|\Phi^+\rangle^{\otimes n} \\
        &= \sum_{x \in \{00, 01, 10, 11\}^n} \mu_x \left(\bigotimes_{i \in [n]} P_{x_i} \otimes I_{2^n}\right)|\Phi^+\rangle^{\otimes n}
    \end{align*}
    Here, $P_{ab} = i^{ab}X^aZ^b$. We then note that $(P_{01}\otimes I)|\Phi^+\rangle = |\Phi^-\rangle, (P_{10}\otimes I)|\Phi^+\rangle = |\Psi^+\rangle, (P_{11}\otimes I)|\Phi^+\rangle = i|\Psi^-\rangle$. Consequently, every one of the $4^n$ Pauli operators corresponds to a unique choice of $n$ Bell states. Since these states are all orthogonal, measuring in the Bell basis is akin to sampling a Pauli $P_x$ with probability $|\mu_x|^2$.
\end{proof}

Empirically, this technique has been successfully utilized in Hamiltonian learning protocols, enabling Heisenberg-limited scaling \cite{hu_ansatz-free_2025}. Furthermore, recent concurrent work has demonstrated its utility for Lindbladian support learning \cite{ivashkov_ansatz-free_2026}. Our work builds on this foundation by leveraging Bell sampling for both support and coefficient learning, eliminating the need to solve linear systems of exponential size in the classical postprocessing step. To see how this can be accomplished, we first extend this procedure to the open system framework, giving us an analogous statement to Fact \ref{fact:bell_sampling_unitary} for quantum channels.

\begin{lemma}\label{lem:bell_sampling_channel}
    For a given quantum channel $\mathcal E[\rho] = \sum_{a, b}\xi_{ab} P_a \rho P_b$, we can sample Pauli $P_k$ with probability $\xi_{kk}$.
\end{lemma}

\begin{proof}
    Performing the same Bell measurement procedure as before, we have the following: 
    \begin{align*}
        \text{Pr}[P_k] &= \langle \Phi_n|(I \otimes P_k)(I \otimes \mathcal E)[|\Phi_n\rangle \langle \Phi_n|](I \otimes P_k)|\Phi_n\rangle \\
        &= \sum_{a, b}\xi_{ab}\langle \Phi_n|(I \otimes P_k)(I \otimes P_a)|\Phi_n\rangle \langle \Phi_n|(I \otimes P_b)(I \otimes P_k)|\Phi_n\rangle \\
        &= \sum_{a, b}\xi_{ab}\langle \Phi_n|(I \otimes P_kP_a)|\Phi_n\rangle \langle \Phi_n|(I \otimes P_bP_k)|\Phi_n\rangle \\
        &= \sum_{a, b}\xi_{ab} \delta_{ak}\delta_{bk} \\
        &= \xi_{kk}
    \end{align*}
\end{proof}

In the context of Hamiltonian learning, Bell sampling can be used to ascertain the support of the Hamiltonian by evolving for short time such that $U(t) = e^{-itH} \approx I - itH$, meaning that any non-identity Pauli obtained via Bell sampling is likely to be present in the Hamiltonian. Extending this intuition to the Lindbladian case, we can see that Bell sampling allows us to measure the support of the terms in the diagonal of $\chi(\mathcal L)$. We formally state this in Corollary \ref{cor:dissipator_supp}.

\begin{corollary}\label{cor:dissipator_supp}
    Given a Lindbladian $\mathcal L = \sum_{a, b}\chi_{ab}P_a\rho P_b$ with traceless jump operators $L_j$, we can sample any Pauli $P_k \in \bigcup_{j}\text{supp}(L_j)$ with probability $\approx \chi_{kk}t$ if $P_k \ne I$, and probability $\approx 1 + \chi_{00} t$ if $P_k = I$.
\end{corollary}
\begin{proof}
    Assume traceless jump operators just as before. Define the Lindbladian in terms of the following Pauli expansion:
    $$\mathcal L[\cdot] = \sum_{a, b}\chi_{ab}P_a \cdot P_b$$
    For small time $t$, we have that
    $$U(t)[\rho] = e^{\mathcal L(t)}[\rho] \approx \rho + t\mathcal L[\rho]$$
    From here, applying Lemma \ref{lem:bell_sampling_channel} and matching terms, we directly get the result stated above.
\end{proof}
Consequently, Bell sampling serves as a powerful tool to both probe the support and also estimate the diagonal coefficients themselves via Monte Carlo experiments. Recalling the discussion in the first part, the support of the diagonal elements directly tells us the support of the dissipative component of the Lindbladian. In particular, we have that defining the diagonal support set $\mathcal S$, for all $i, j \in \mathcal S$, the only nonzero terms with contributions from the dissipator in $\chi(\mathcal L)$ are $\chi_{ij}$ and $\text{Re}[\chi_{0k}] = \text{Re}[\chi_{k0}]$, where $P_k \propto P_iP_j$.
Hence, when we adapt our protocol to learn the off-diagonal components of $\chi(\mathcal L)$, we can restrict the number of coefficients we need to learn to a quantity that is quadratic in $M$ (up to logarithmic factors as we will later see). 

Nevertheless, while this method allows us to completely characterize the support of the dissipator, we are left with little information about the Hamiltonian support $\mathcal S_H$, corresponding to the imaginary components of the elements in the top row and leftmost column of $\chi(\mathcal L)$, i.e. $\text{Im}[\chi_{0j}] = -\text{Im}[\chi_{j0}]$. Since the Hamiltonian is defined independently from the jump operators, and the on-diagonal contributions of the Hamiltonian are only observable in the second-order terms, we need to adjust our sampling procedure accordingly to ensure that we can also find these terms. To simplify matters, we can assume that we will try to estimate potential Hamiltonian contributions for any of the terms $\text{Re}[\chi_{0k}] = \text{Re}[\chi_{k0}]$ mentioned previously which arise from the dissipator term.

Consequently, we only concern ourselves with separately learning $\mathcal S_H' = \mathcal S_H \setminus (\mathcal S_H \cap \mathcal S')$, where $\mathcal S' = \{a \oplus b: a, b \in \mathcal S\}$. Since there are no quadratic contributions from the dissipator term for the $\mathcal S_H$ Bell sampling probabilities, the quadratic order contribution is dominated by the actual Hamiltonian term itself. In particular, we have the following lemma:
\begin{lemma}\label{lem:hamiltonian_supp}
    For any Pauli $P_m \in \mathcal S_H'$, we can sample these Paulis with probability $\approx \text{Im}[\chi_{m0}]^2t^2$ for small $t$.
\end{lemma}

\begin{proof}
    As we saw previously, expanding to linear order ignores the contribution of the linear term. Hence, we must expand $e^{\mathcal Lt}$ to quadratic order to get the Hamiltonian contribution. Since we are only considering those terms in the Hamiltonian which cannot be formed from Paulis in the support of the jump operators or their pairwise products, we have that for $P_m \in \mathcal S_H'$, the only quadratic order term $P_m|\Psi\rangle\langle \Psi|P_m$ arises from the Hamiltonian term itself. Hence, we have that for a Hamiltonian $H = \sum_j h_jP_j$, we have that
    $$\Pr[P_m] \approx h_m^2t^2 = \text{Im}[\chi_{0m}]^2t^2$$
\end{proof}

Since we will be evolving for time $t \ll 1$, these quadratic terms will be somewhat more challenging to sample compared to the linear order dissipator support terms. Consequently, Hamiltonian support learning will prove to be the bottleneck for this protocol, as we will later see. To determine how much overhead this incurs, we need to determine the optimal evolution time $t$ to choose to minimize the overall time evolution. In general, prolonged time evolution increases the likelihood of sampling the desired outcomes, forming the cornerstone of sample-efficient learning approaches \cite{ivashkov_ansatz-free_2026}. Consequently, choosing longer evolution times often results in less evolution time overall. However, this also introduces higher-order terms in the time evolution, which complicates the analysis of the sampling probability bounds. Picking an optimal time, therefore, requires finding the longest interval for which we can evolve without incurring Taylor error larger than the order of the smallest probability we need to measure.

In Appendix \ref{chap:support_learn_proofs}, we show that the Taylor series error for the linear order probabilities can be bounded as $\varepsilon_\text{Taylor} \le \mathcal O(M^2t^2)$
for small $t$. This intuitively follows from the fact that there are at most $\mathcal O(M)$ pairs of Paulis in $\mathcal S$ whose product is proportional to a given Pauli in $\mathcal S$. Consequently, since there are independent contributions on the left and right sides, we have contributions from $\mathcal O(M^2)$ terms in the second-order expansion. Since we need this error to be on the order of the smallest probability in the linear order sampling regime, which happens to be $\mathcal O(\epsilon t)$, this forces us to choose $t = \mathcal O(\epsilon/M^2)$, meaning the smallest probability is on the order of $\mathcal O(\epsilon^2/M^2)$. Performing concentration analysis via Hoeffding's inequality, this naturally gives us a sample complexity of $N = \widetilde{\mathcal O}(M^2/\epsilon^2)$ and time evolution of $T = \widetilde{\mathcal O}\left(1/\epsilon\right)$ here. In other words, we can learn the dissipator support with Heisenberg-limited scaling up to logarithmic factors. Altogether, we write the dissipator support learning protocol as follows.

\begin{algorithm}[H]
\caption{Nonlocal Dissipator Support Learning}\label{alg:dissipator_support_learn}
\begin{algorithmic}[1]
\REQUIRE Query access to $\mathcal E_t[\cdot] = e^{t \mathcal L}[\cdot]$, accuracy threshold $\epsilon$, sparsity bound $M$
\STATE Set $\hat {\mathcal S} \leftarrow \emptyset$
\STATE Set $t \leftarrow \mathcal O(\epsilon/M^2)$.
\STATE Set $N \leftarrow \mathcal O\left(\frac{M^2\log(M/\delta)}{\epsilon^2}\right)$ 
\FOR{$i = 1, \dots, N$}
    \STATE Sample Pauli $P_j$ from $\mathcal E_t$ via Bell sampling
    \STATE Update $\hat {\mathcal S} \leftarrow \hat {\mathcal S} \cup j$
\ENDFOR
\RETURN $\hat {\mathcal S}$
\ENSURE Candidate dissipator support set $\hat {\mathcal S}$
\end{algorithmic}
\end{algorithm}

\begin{remark}
    While we denote the set $\hat {\mathcal S}$ returned here as the "dissipator support", we note that $\hat {\mathcal S}$ does not technically cover the entire support of the dissipator $\mathcal L_D$. In particular, the anticommutator terms may produce Paulis that lie outside $\hat {\mathcal S}$. Nevertheless, we can always recover these terms using the fact that they all take the form of $P_k \rho$ or $\rho P_k$, were $P_k \propto P_iP_j$ for Paulis $P_i, P_j \in \hat {\mathcal S}$.
\end{remark}

For the Hamiltonian terms, we instead perform analysis on the third-order term, where a Taylor error bound of $\varepsilon_\text{Taylor}^H \le \mathcal O(M^4t^3)$ forces us to choose $t \le \mathcal O(\epsilon^2/M^4)$ since the smallest probability is on the order of $\mathcal O(\epsilon^2t^2) = \mathcal O(\epsilon^6/M^4)$. Choosing a smaller evolution time for the Hamiltonian terms may seem antithetical given that these terms are harder to sample; however, due to the potential for the third-order contribution to significantly perturb the probability of sampling in an uncontrollable fashion, we are ultimately compelled to choose a shorter evolution time to ensure that the Taylor error is properly controlled. Hence, we instead have a sample complexity of $N_H = \widetilde{\mathcal O}(M^8/\epsilon^6)$ and an evolution time of $T_H = \widetilde{\mathcal O}(M^4/\epsilon^4)$. This then gives us the following protocol for learning the Hamiltonian support.

\begin{algorithm}[H]
\caption{Nonlocal Hamiltonian Support Learning}\label{alg:hamiltonian_support_learn}
\begin{algorithmic}[1]
\REQUIRE Query access to $\mathcal E_t[\cdot] = e^{t \mathcal L}[\cdot]$, accuracy threshold $\epsilon$, sparsity bound $M$, dissipator support $\hat {\mathcal S}$
\STATE Set $\hat{\mathcal S}_H \leftarrow \{a \oplus b : a, b \in \hat{\mathcal S}\}$
\STATE Set $t_H \leftarrow \mathcal O(\epsilon^2/M^4)$
\STATE Set $N_H \leftarrow \mathcal O\left(\frac{M^8\log(M/\delta)}{\epsilon^6}\right)$
\FOR{$i = 1, \dots, N_H$}
    \STATE Sample Pauli $P_j$ from $\mathcal E_{t_H}$ via Bell sampling
    \STATE Update $\hat {\mathcal S}_H \leftarrow \hat{\mathcal S}_H \cup j$
\ENDFOR
\RETURN $\hat {\mathcal S}_H$
\ENSURE Candidate Hamiltonian support set $\hat {\mathcal S}_H$
\end{algorithmic}
\end{algorithm}

Putting these results together, we obtain the following theorem and corollary:
\begin{restatable}[Nonlocal Lindbladian Support Learning]{proposition}{propsupportlearn}\label{prop:support_learn}
    Given an $M$-sparse Lindbladian $\mathcal L$ with dissipator and Hamiltonian support $\mathcal S, \mathcal S_H$, we can learn sets $\hat {\mathcal S}, \hat {\mathcal S_H}$ such that $\mathcal S \subseteq \hat {\mathcal S}, \mathcal S_H \subseteq \hat {\mathcal S_H}$ with high probability in total evolution time
    $$T = {\mathcal O}\left(\frac{M^4 \log M}{\epsilon^4}\right)$$
\end{restatable}
\begin{restatable}[Nonlocal Lindbladian Dissipator Support Learning]{corollary}{cordissipatorsupportlearn}\label{cor:dissipator_support_learn}
    Given a purely dissipative $M$-sparse Lindbladian $\mathcal L$ with support $\mathcal S$, we can learn $\hat {\mathcal S}$ such that $\mathcal S \subseteq \hat {\mathcal S}$ with high probability in total evolution time
    $$T = {\mathcal O}\left(\frac{\log M}{\epsilon}\right)$$
\end{restatable}
Note that here, we assume that probability of success $\delta$ is some constant very close to zero. The proof of correctness for these theorems follows trivially from applying Corollary \ref{cor:dissipator_supp} and Lemma \ref{lem:hamiltonian_supp} and the subsequent error analysis. Consequently, $\mathcal S \subseteq \hat {\mathcal S}$ and $\mathcal S_H \subseteq \hat {\mathcal S_H}$ with high probability. Nevertheless, this also implies that our algorithm is likely to oversample the support sets. In order to bound how large these sets are, we require Lemma \ref{lem:support_size}.

\begin{restatable}{lemma}{lemsupportsize}\label{lem:support_size}
    Let $\hat {\mathcal S}, \hat {\mathcal S}_H$ be the support sets returned by Algorithms \ref{alg:dissipator_support_learn} and \ref{alg:hamiltonian_support_learn}. We then have that $|\hat {\mathcal S}| \le \mathcal O(M \log M)$ and $|\hat {\mathcal S}_H| \le |\hat {\mathcal S}|^2 + \mathcal O(M \log M)$, respectively, with high probability.
\end{restatable}

The full proof of this lemma is provided in Appendix \ref{chap:support_learn_proofs}. Intuitively, however, this result arises from an application of the multiplicative Chernoff bound on the probability of sampling outside the support of the dissipator and Hamiltonian, which is then bounded by the Taylor error analyzed earlier. Ultimately, this lemma affirms that we only oversample the dissipator and Hamiltonian support sets by a polylogarithmic factor of $M$, meaning that the mismatch in size does not introduce any significant overhead in terms of the number of coefficients we have to estimate.

\section{Learning Coefficients}

Now that we have learned the support, we are ready to learn the actual coefficients of the Lindbladian. Previous approaches to coefficient learning relied on various tomography methods such as expectation value derivative estimation \cite{stilck_franca_efficient_2024, ivashkov_ansatz-free_2026} to recover the coefficients. However, if the Paulis in the Lindbladian are highly nonlocal, this potentially requires solving a linear system of exponential size. Consequently, we aim to derive a coefficient learning algorithm that is both efficient in evolution time and avoids potentially intractable classical post-processing.

\subsection{Learning the Diagonal Elements}

From Corollary \ref{cor:dissipator_supp}, we have that Bell sampling allows us to probe the support of the entire Lindbladian by sampling with probability proportional to the Chi matrix coefficient times the time we evolve for. Consequently, by measuring the frequency with which we sample, we can estimate these coefficients. We previously observed that by evolving for time $t = \mathcal O(\epsilon/M^2)$, the probability of sampling a given diagonal coefficients $p_{kk}$ differs from $\chi_{kk}t$ by at most $\mathcal O(\epsilon^2/M^2)$. Hence, if this error is suppressed by a sufficiently large constant factor, we can learn $\chi_{kk}$ to $\epsilon$-accuracy with enough samples. Applying the Hoeffding bound, we have that estimating the probability to accuracy $\mathcal O(\epsilon^2/M^2)$ with high probability requires $N = \widetilde{\mathcal O}(M^4/\epsilon^4)$ samples, giving us Proposition.

\begin{restatable}[Complexity of On-Diagonal Coefficient Learning]{proposition}{propondiagonal}\label{prop:on_diagonal}
    Given an $M$-sparse Lindbladian $\mathcal L$ with dissipator support $\hat {\mathcal S}$, we can learn the diagonal components of $\chi(\mathcal L)$ in total evolution time
    $$T = {\mathcal O}\left(\frac{M^2 \log M}{\epsilon^3}\right)$$
\end{restatable}

\begin{remark}
    While the evolution time bound of $T = \widetilde{\mathcal O}\left(M^2/\epsilon^3\right)$ presented in Proposition \ref{prop:on_diagonal} is sufficient for our purposes, it should be noted that this bound can be optimized further. In particular, if we frequently twirl the Lindblad evolution to isolate the diagonal components, we can use the eigenvalue estimation procedure defined in Chapter \ref{chap:spam_lindblad_learn} paired with knowledge of the support via Bell sampling to remove a factor of $1/\epsilon$ without introducing exponential classical overhead. However, because this cannot be used to improve the runtime of off-diagonal learning, and it does not change the overall runtime of the Lindbladian learning algorithm due to the complexity of learning the Hamiltonian support, we utilize this simpler bound for the remainder of this chapter.
\end{remark}

It should also be noted that during the coefficient estimation procedure, we will estimate coefficients for the Paulis we oversampled in the support. Consequently, we can discard these Paulis by checking if they fall below some coefficient threshold. For instance, we can simply choose the parameters accordingly so that we learn the off-diagonals to $\epsilon/2$-accuracy instead of $\epsilon$-accuracy, meaning that any Pauli not found in the support will have a learned coefficient below $\epsilon/2$ with high probability. As a small caveat, it should also be noted that while the sampling probability is given approximately as $\chi_{kk}t$ for $k \ne 0$, we have that the probability of sampling the identity is approximated as $1 + \chi_{00}t = 1 - t\sum_{k \ne 0}\chi_{kk}$. Writing the on-diagonal learning algorithm in full detail, we have the following procedure:
\begin{algorithm}[H]
\caption{Nonlocal On-Diagonal Lindbladian Coefficient Learning}\label{alg:on_diag_coeff_learn}
\begin{algorithmic}[1]
\REQUIRE Query access to $\mathcal E_t[\cdot] = e^{t \mathcal L}[\cdot]$, accuracy threshold $\epsilon$, sparsity bound $M$, learned dissipator support $\hat {\mathcal S}$
\STATE Initialize all $\hat \chi_{kk}, p_{k} = 0$.
\STATE Set $t \leftarrow \mathcal O(\epsilon/M^2)$
\STATE Set $N \leftarrow \widetilde{\mathcal O}\left(M^4/\epsilon^4\right)$ 
\FOR{$\ell = 1, \dots, N$}
    \STATE Sample Pauli $P_k$ from $\mathcal E_t$ via Bell sampling
    \STATE If $k \in \hat {\mathcal S}$, update $p_{k} \leftarrow p_{k} + 1/N$
\ENDFOR
\STATE Set $\hat \chi_{00} \leftarrow (p_0 - 1)/t$
\FOR{$k \in \hat {\mathcal S} \setminus \{0\}$}
    \STATE Set $\hat \chi_{kk} \leftarrow p_{k}/t$
    \IF{$\hat \chi_{kk} < \frac{\epsilon}{2}$}
        \STATE Reset $\hat \chi_{kk} \leftarrow 0$ and remove $k$ from $\hat {\mathcal S}$
    \ENDIF
\ENDFOR
\RETURN $\hat \chi_{kk}, \hat {\mathcal S}$
\ENSURE Nonzero learned coefficients $\hat \chi_{kk}$ and new dissipator support $\hat {\mathcal S}$
\end{algorithmic}
\end{algorithm}
Since this algorithm modifies the learned support to give us the true dissipator support with high probability, we have that the size of $|\hat {\mathcal S}| = \mathcal O(M)$. Consequently, by Lemma \ref{lem:support_size}, we have that $|\hat {\mathcal S}_H| \le \mathcal O(M^2)$.

\subsection{Learning the Off-Diagonals}

Now that we have learned the diagonal elements of $\chi(\mathcal L)$, we now focus on learning the rest of the $\chi$ matrix. To accomplish this, we need to mix the components of the $\chi$ matrix such that the off-diagonal elements appear on the diagonal. We can do this by applying Clifford operators before performing the Bell basis measurement. In particular, consider the following sets of operators shown in Table \ref{tab:transformation}.
\begin{table}[H]
    \centering
    \begin{tabular}{|c|c|c|c|} \hline
        Transformation & $[P_i, P_j] = 0$? & Unitary & Component Learned \\ \hline
        $P_i \pm P_j$ & N & Y & $\text{Re}[\chi_{ij}]$ \\
        $P_i \pm iP_j$ & Y & Y & $\text{Im}[\chi_{ij}]$ \\
        $P_i \pm iP_j$ & N & N & $\text{Im}[\chi_{ij}]$ \\
        $P_i \pm P_j$ & Y & N & $\text{Re}[\chi_{ij}]$ \\ \hline
    \end{tabular}
    \caption{Types of transformations required to extract the real and imaginary components of the cross-terms.}
    \label{tab:transformation}
\end{table}
The first two transformations are proportional to Clifford operators. In particular, if we define
\begin{equation}
    C_{ij} = \begin{cases}
        \frac{P_i + P_j}{\sqrt 2} & \{P_i, P_j\} = 0 \\
        \frac{P_i + iP_j}{\sqrt 2} & [P_i, P_j] = 0
    \end{cases},
\end{equation}
we can prove the following result.
\begin{lemma}\label{lem:cliff}
    The operators $C_{ij}$ are Clifford operators.
\end{lemma}
\begin{proof}
    Choosing $P_i = X, P_j = Z$ gives us $C_{ij} = H$, which is trivially Clifford. Furthermore, choosing $P_i = XX, P_j = YY$ gives us $C_{ij} \propto (X \otimes X)(S \otimes S)\text{CZ}$. Since every pair of anticommuting Paulis can be mapped to every other pair of anticommuting Paulis via Clifford operators, and likewise for commuting Paulis, this means that any operator of the prescribed form is a Clifford.
\end{proof}

For the other two non-unitary transformations, we can implement these using ancilla qubits and post-selection. In particular, suppose we define the Kraus operators $K_{ij}^\pm$ as follows:
\begin{equation}
    K_{ij}^\pm = \begin{cases}
    \frac{P_i \pm iP_j}{2} & \{P_i, P_j\} = 0 \\
    \frac{P_i \pm P_j}{2} & [P_i, P_j] = 0
\end{cases}
\end{equation}
With these operators, we then have the following lemma.

\begin{lemma}\label{lem:kraus}
    The operators $K_{ij}^\pm$ form a complete set of Kraus operators.
\end{lemma}

\begin{proof}
    Working out the commuting and anticommuting cases for $K_{ij}^\pm$, we have that
    \begin{equation}\label{eq:kraus_normalization}
        (K_{ij}^\pm)^\dag K_{ij}^\pm = \begin{cases}
        \frac{I \pm iP_iP_j}{2} & \{P_i, P_j\} = 0 \\
        \frac{I \pm P_iP_j}{2} & [P_i, P_j] = 0
    \end{cases}
    \end{equation}
    Hence, in both the anticommuting and commuting cases, we have that
    $$(K_{ij}^+)^\dag K_{ij}^+ + (K_{ij}^-)^\dag K_{ij}^- = I$$
\end{proof}

From Lemma \ref{lem:kraus}, we can define a Stinespring dilation $\Phi_{ij}$ that maps
$$|0\rangle \otimes |\psi\rangle \mapsto |0\rangle \otimes K_{ij}^+|\psi\rangle + |1\rangle \otimes K_{ij}^-|\psi\rangle$$
Consequently, by performing this isometry and measuring the ancilla qubit, we can effectively apply these non-unitary transformations to the system.

To demonstrate how these transformations affect the Lindbladian coefficients, we need only consider terms of the form
\begin{equation}
    \mathcal L[\rho] = \chi_{ii}P_i\rho P_i + \chi_{ij}P_i\rho P_j + \chi_{ji}P_j\rho P_i + \chi_{jj}P_j\rho P_j + \dots
\end{equation}
This is due to the fact that these are the only terms that can be mapped to diagonal terms of the form $\rho$ and $P_iP_j \rho P_j P_i$, which can be probed via Bell sampling. Defining $P_k \propto P_iP_j$, we have that these transformations give us the following new coefficients $\chi'_{00}$ and $\chi'_{kk}$ in the Chi matrix expansion for $\mathcal L' = \mathcal T \circ \mathcal L$, where $\mathcal T$ is any one of the four transformations specified in Table \ref{tab:transformation}. Altogether, this gives us the following table:
\begin{table}[H]
    \centering
    \begin{tabular}{|c|c|c|c|} \hline
        $\mathcal T$ & $[P_i, P_j] = 0$? & $\chi_{00}'$ & $
        \chi_{kk}'$ \\ \hline
        $C_{ij} \cdot C_{ij}^\dag$ & N & $\frac{1}{2}(\chi_{ii} + \chi_{jj} + 2\text{Re}[\chi_{ij}])$ & $\frac{1}{2}(\chi_{ii} + \chi_{jj} - 2\text{Re}[\chi_{ij}])$ \\
        $C_{ij} \cdot C_{ij}^\dag$ & Y & $\frac{1}{2}(\chi_{ii} + \chi_{jj} + 2\text{Im}[\chi_{ij}])$ & $\frac{1}{2}(\chi_{ii} + \chi_{jj} - 2\text{Im}[\chi_{ij}])$ \\
        $K_{ij}^\pm \cdot (K_{ij}^\pm)^\dag$ & N & $\frac{1}{4}(\chi_{ii} + \chi_{jj} \pm 2\text{Im}[\chi_{ij}])$ & $\frac{1}{4}(\chi_{ii} + \chi_{jj} \pm 2\text{Im}[\chi_{ij}])$ \\
        $K_{ij}^\pm \cdot (K_{ij}^\pm)^\dag$ & Y & $\frac{1}{4}(\chi_{ii} + \chi_{jj} \pm 2\text{Re}[\chi_{ij}])$ & $\frac{1}{4}(\chi_{ii} + \chi_{jj} \pm 2\text{Re}[\chi_{ij}])$\\ \hline
    \end{tabular}
    \caption{Coefficients of Chi matrix expansion of $\mathcal T \circ \mathcal L$ for terms $\rho$ and $P_k\rho P_k$. Note that for the Kraus operators, the transformations shown above correspond to the unnormalized coefficients. To renormalize these terms, we have to divide by $\kappa_{ij}^\pm = \text{Tr}(K^\pm_{ij} e^{\mathcal Lt}[\rho] (K^\pm_{ij})^\dag))$, which corresponds to the probability of measuring ancilla 0 in the case of $+$ and measuring ancilla 1 in the case of $-$.}
    \label{tab:transformation_2}
\end{table}

As we can see, these transformations give us linear combinations of coefficients from which we can neatly recover the off-diagonal coefficients. In particular, for the Clifford transformations, we can simply estimate the real and imaginary components of $\chi_{ij}$ by taking $\frac{\chi_{00}' - \chi_{kk}'}{2}$. 

However, for the Kraus transformations, extracting the off-diagonal components requires extra computation. Due to the non-unitary nature of the transformations, we have to renormalize the state by a factor of $1/\kappa_{ij}^\pm$, where $\kappa_{ij}^\pm = \text{Tr}(K^\pm_{ij} e^{\mathcal Lt}[\rho] (K^\pm_{ij})^\dag))$ corresponds to the probability of applying $K_{ij}^+$ and $K_{ij}^-$, respectively. In Appendix \ref{chap:support_learn_proofs}, we show that this normalization factor reduces neatly to either $\frac{1 + \langle P_k\rangle}{2}$ or $\frac{1 - \langle P_k\rangle}{2}$, from which we then arrive at the following lemma:

\begin{restatable}{lemma}{lemnormalization}\label{lem:normalization}
    The normalization factors $\kappa_{ij}^\pm = \text{Tr}(K^\pm_{ij} e^{\mathcal Lt}[\rho] (K^\pm_{ij})^\dag))$ with $t = \mathcal O(\epsilon/M^2)$ satisfy the following bound with high probability:
    $$\left|\kappa_{ij}^\pm - \frac{1}{2}\right| \le \mathcal O(\epsilon)$$
\end{restatable}
This lemma implies that $\kappa_{ij}^\pm$ differs from $1/2$ by at most $\mathcal O(\epsilon)$, meaning that we will sample each Kraus operator with almost equal probability. This implies that we can estimate the sampling probabilities for both Kraus operators to similar precision for each. After we rescale the estimated probabilities by $\kappa_{ij}^\pm$, we can then take $\chi'_{00} + \chi'_{kk}$ for both the $\pm$ cases, and then take their average deviation to recover the real and imaginary components.

Further error analysis on $\kappa_{ij}^\pm$ in Appendix \ref{chap:support_learn_proofs} elucidates that the sampling complexity $N = \widetilde{\mathcal O}(M^4/\epsilon^4)$ utilized in the on-diagonal case is sufficient for both estimating these normalization coefficients to the desired accuracy and estimating the coefficients themselves to $\mathcal O(\epsilon)$-accuracy. Since this is the sample complexity for a single off-diagonal component, and we have $\mathcal O(M^2)$ off-diagonal elements to learn at most, we have a sample complexity of ${\mathcal O}(M^6/\epsilon^4)$. Hence, we have the following proposition:

\begin{restatable}[Complexity of Off-Diagonal Coefficient Learning]{proposition}{propoffdiagonal}\label{prop:off_diagonal}
    Given an $M$-sparse Lindbladian $\mathcal L$ with dissipator support $\hat {\mathcal S}$ and Hamiltonian support $\hat {\mathcal S}_H$, we can learn all the nonzero off-diagonal components of $\chi(\mathcal L)$ in total evolution time
    $$T = {\mathcal O}\left(\frac{M^4}{\epsilon^3}\right)$$
\end{restatable}

Notably, for coefficient learning, we did not have to treat the Hamiltonian terms any different from the others, as the unitary technique mentioned previously inserts $\text{Im}[\chi_{0k}]$ into the linear order diagonal terms. Altogether, we have the following two algorithms:
\begin{algorithm}[H]
\caption{Lindblad Off-Diagonal Coefficient Learning via Cliffords}\label{alg:off_diag_coeff_learn_cliff}
\begin{algorithmic}[1]
\REQUIRE Query access to $\mathcal E_t[\cdot] = e^{t \mathcal L}[\cdot]$, accuracy threshold $\epsilon$, sparsity bound $M$, learned coefficient support $\hat {\mathcal S}'$
\STATE Initialize $\hat \chi_{ab} = 0$ for all $(a, b) \in \mathcal {\hat S'}$ where $a \ne b$
\STATE Set $t \leftarrow \mathcal O(\epsilon/M^2)$
\STATE Set $N \leftarrow {\mathcal O}\left(M^4/\epsilon^4\right)$
\FOR{$a, b \in \hat {\mathcal S}'$ s.t. $a < b$}
    \STATE Set $p_1, p_2 \leftarrow 0$
    \FOR{$\ell = 1, \dots, N$}
        \STATE Sample $P_m$ from $\mathcal C_{ij} \circ \mathcal E_t$ via Bell sampling
        \IF{$P_m = I$}
            \STATE Update $p_1 \leftarrow p_1 + 1/N$
        \ELSIF{$P_m \propto P_aP_b$}
            \STATE Update $p_2 \leftarrow p_2 + 1/N$
        \ENDIF
    \ENDFOR
    \STATE Set $\hat \chi_{ab} \leftarrow \frac{i^{\langle a, b \rangle}}{2t}(p_1 - p_2)$
    \STATE Set $\hat \chi_{ba} = \hat \chi_{ab}^*$
\ENDFOR
\RETURN $\hat \chi_{ab}$
\ENSURE Partially learned off-diagonal coefficients $\hat \chi_{ab}$
\end{algorithmic}
\end{algorithm}

\begin{algorithm}[H]
\caption{Lindblad Off-Diagonal Coefficient Learning via Isometries}\label{alg:off_diag_coeff_learn_iso}
\begin{algorithmic}[1]
\REQUIRE Query access to $\mathcal E_t[\cdot] = e^{t \mathcal L}[\cdot]$, accuracy threshold $\epsilon$, sparsity bound $M$, learned coefficient support $\hat {\mathcal S}''$
\STATE Initialize $\hat \chi_{ab} = 0$ for all $(a, b) \in \mathcal {\hat S}''$ where $a \ne b$
\STATE Set $t \leftarrow \mathcal O(\epsilon/M^2)$
\STATE Set $N \leftarrow \widetilde {\mathcal O}\left(M^4/\epsilon^4\right)$
\FOR{$a, b \in \hat {\mathcal S}''$ s.t. $a < b$}
    \STATE Set $p^{(0)}, p^{(1)}, \kappa_{ab} \leftarrow 0$
    \FOR{$\ell = 1, \dots, N$}
        \STATE Sample $P_m$ from $\Phi_{ij} \circ (I \otimes \mathcal E_t)$ via Bell sampling with ancilla output $c$
        \STATE Update $\kappa_{ab} \leftarrow \kappa_{ab} + c/N$
        \IF{$P_m = I$ or $P_m \propto P_aP_b$}
            \IF{$c = 0$}
                \STATE Update $p^{(0)} \leftarrow p^{(0)} + 1/N$
            \ELSE
                \STATE Update $p^{(1)} \leftarrow p^{(1)} + 1/N$
            \ENDIF
        \ENDIF
    \ENDFOR
    \STATE Set $\hat \chi_{ab} \leftarrow \frac{i^{1 - \langle a, b \rangle}}{2t}((1 - \kappa_{ab})p^{(0)} - \kappa_{ab}p^{(1)})$
    \STATE Set $\hat \chi_{ba} = \hat \chi_{ab}^*$
\ENDFOR
\RETURN $\hat \chi_{ab}$
\ENSURE Partially learned off-diagonal coefficients $\hat \chi_{ab}$
\end{algorithmic}
\end{algorithm}

\section{Full Noiseless Learning Algorithm}

Combining all the support learning and coefficient learning subroutines, we obtain the following algorithm:
\begin{algorithm}[H]
\caption{Nonlocal Sparse Lindbladian Learning}\label{alg:lindblad_learn}
\begin{algorithmic}[1]
\REQUIRE Query access to $\mathcal E_t[\cdot] = e^{t \mathcal L}[\cdot]$, accuracy threshold $\epsilon$, sparsity bound $M$
\STATE Run Algorithm \ref{alg:dissipator_support_learn} to recover the dissipator support set $\hat{\mathcal S}$
\STATE Run Algorithm \ref{alg:on_diag_coeff_learn} to learn $\hat \chi_{kk}$ and recover the trimmed dissipator support set $\hat{\mathcal S}$
\STATE Run Algorithm \ref{alg:hamiltonian_support_learn} to recover the Hamiltonian support set $\hat{\mathcal S}_H$
\STATE Define support sets $\hat {\mathcal S}' = $$(\hat {\mathcal S} \times \hat{\mathcal S}) \cup (\{0\} \times \hat{\mathcal S}_H) \cup (\hat{\mathcal S}_H \times \{0\})$ and $\hat {\mathcal S}'' = $$(\hat {\mathcal S} \times \hat{\mathcal S}) \cup (\{0\} \times \hat{\mathcal S}^2) \cup (\hat{\mathcal S}^2 \times \{0\})$
\STATE Run Algorithm \ref{alg:off_diag_coeff_learn_cliff} to get coefficients $\hat \chi^{(1)}_{ab}$
\STATE Run Algorithm \ref{alg:off_diag_coeff_learn_iso} to get coefficients $\hat \chi^{(2)}_{ab}$
\STATE For $a \ne b$ s.t. $(a, b) \in \hat{\mathcal S}''$, set $\hat \chi_{ab} \leftarrow \hat \chi^{(1)}_{ab} + \hat \chi^{(2)}_{ab}$
\STATE For $a \ne b$ s.t. $(a, b) \in \hat{\mathcal S}' \setminus \hat{\mathcal S}''$, set $\hat \chi_{ab} \leftarrow \hat \chi^{(1)}_{ab}$
\RETURN $(\hat \chi_{ab})$
\ENSURE Learned Chi matrix $(\hat \chi_{ab})$ for Lindbladian $\mathcal L$
\end{algorithmic}
\end{algorithm}

Furthermore, combining all the evolution times for each of these subroutines, we obtain the result posited at the beginning of the chapter:
\begin{theorem}
    Given an $M$-sparse Lindbladian $\mathcal L[\rho] = \sum_{ij} \chi_{ij} P_i \rho P_j$, we can learn coefficients $\hat \chi_{ij}$ s.t.
    $$|\hat \chi_{ij} - \chi_{ij}| \le \epsilon$$
    with high probability for all $i, j$ in total evolution time
    $$T = {\mathcal O}\left(\frac{M^4 \log M}{\epsilon^4}\right)$$
    using Algorithm \ref{alg:lindblad_learn}.
\end{theorem}
The proof of this theorem follows trivially by combining the results of Propositions \ref{prop:support_learn}, \ref{prop:on_diagonal}, and \ref{prop:off_diagonal}. This result matches the $\epsilon$-dependence of the evolution time posited in \cite{ivashkov_ansatz-free_2026}, allowing us to efficiently learn the Lindbladian without requiring exponential classical postprocessing at the cost of only an additional factor of $\widetilde{\mathcal O}(M)$.


\chapter{SPAM-Robust Lindbladian Learning}\label{chap:spam_lindblad_learn}

While the previous algorithm operated under noiseless assumptions, in practice, noise is an unavoidable part of quantum devices, and must therefore be properly accounted for. In particular, SPAM error poses a significant risk by making state preparation and measurement operations unreliable. Since primitives like Bell sampling require frequent measurements, these methods are infeasible to implement on our current devices. Consequently, in this chapter, we will investigate how to mitigate the effect of SPAM error on the learning process, allowing us to learn much of the Lindbladian without any hindrance. For our purposes, we will assume the following bound on the SPAM error, as defined in \cite{hu_ansatz-free_2025}:
\begin{definition}[SPAM Error]\label{def:spam}
    We define the state preparation noise as an error channel $\mathcal E_\text{prep}$ applied after the ideal state preparation channel and the measurement error as an error channel $\mathcal E_\text{meas}$ applied before the ideal measurement channel. We then bound the distance between these channels and their ideal counterparts as
    $$\| \mathcal E_\text{prep} - \mathcal I \|_\diamond + \| \mathcal E_\text{meas} - \mathcal I \|_\diamond \le \varepsilon_\text{SPAM},$$
    where $\varepsilon_\text{SPAM}$ is the bound on the SPAM error.
\end{definition}
Having defined the SPAM error threshold $\varepsilon_\text{SPAM}$, we can now formulate the problem we aim to solve in this chapter:
\begin{problem}
    Given a SPAM error threshold $\varepsilon_\text{SPAM} \sim \mathcal O(1)$, what components of the Lindbladian Chi matrix are possible to efficiently learn to arbitrary accuracy? 
\end{problem}

To answer this question, we provide the following three central results:
\begin{result}[Unlearnability]
    There exist $n$-qubit Lindbladians $\mathcal L$ such that the components $\text{Re}[\chi_{ij}]$ for $[P_i, P_j] = 0$ and $\text{Im}[\chi_{ij}]$ for $\{P_i, P_j\} = 0$ that cannot be learned to arbitrary accuracy.
\end{result}
\begin{result}[SPAM-Robust Lindbladian Learning]
    Given an $M$-sparse Lindbladian $\mathcal L$, we can estimate the learnable components of $\chi(\mathcal L)$ to $\epsilon$-accuracy in time
    $$T = \widetilde{\mathcal O}\left(\frac{M^{13}n}{\epsilon^8}\right)$$
\end{result}
\begin{result}[SPAM-Robust Lindbladian Learning with Known Support]
    Given an $M$-sparse Lindbladian $\mathcal L$ with known support set $\mathcal S$, we can estimate the learnable components of $\chi(\mathcal L)$ to $\epsilon$-accuracy in time
    $$T = \widetilde{\mathcal O}\left(\frac{M^{5}n}{\epsilon^3}\right)$$
    with efficient classical postprocessing.
\end{result}

\section{Gauge Degrees of Freedom}

To characterize the unlearnable components of the Lindbladian, we employ the notion of a gauge. In quantum learning theory, there often exist parameterized transformations of the operators in any experimental setup that yield identical experimental outcomes, meaning that it is impossible to exactly characterize the nature of these operators in isolation. The parameter that controls these transformations is referred to as a gauge degree of freedom. To define a gauge degree of freedom in the Lindbladian learning context, we consider the most general form of an experiment to learn a Lindbladian $\mathcal L$, consisting of the following sequence:
\begin{equation}\label{eq:experiment}
    (\mathcal M \circ \mathcal E_\text{meas} \circ \mathcal U_N \circ e^{t_N\mathcal L} \circ \cdots \circ \mathcal U_1 \circ e^{t_1\mathcal L} \circ \mathcal E_\text{prep})\rho
\end{equation}
where each $\mathcal U_i[\cdot] = U_i \cdot U_i^\dag$ corresponds to a unitary operation and $\mathcal M$ corresponds to the ideal measurement channel. We define a gauge degree of freedom $\theta$ to be a variable such that there exists some parameterization $\mathcal L' = \mathcal L(\theta), 
\mathcal E_\text{prep}' = \mathcal E_\text{prep}(\theta), \mathcal E_\text{meas}' = \mathcal E_\text{meas}(\theta)$ such that the experiment
\begin{equation}\label{eq:gauge_experiment}
    (\mathcal M \circ \mathcal E_\text{meas}' \circ \mathcal U_N \circ e^{t_N\mathcal L'} \circ \cdots \circ \mathcal U_1 \circ e^{t_1\mathcal L} \circ \mathcal E'_\text{prep})\rho,
\end{equation}
is indistinguishable from that shown previously, with both the error channels still satisfying Definition \ref{def:spam}:
$$\| \mathcal E'_\text{prep} - \mathcal I \|_\diamond + \| \mathcal E'_\text{meas} - \mathcal I \|_\diamond \le \varepsilon_\text{SPAM},$$
By identifying gauge degrees of freedom for certain classes of Lindbladians, we determine the coefficients of the Chi matrix that are gauge dependent and therefore cannot be learned in general.

\subsection{Single-Qubit Lindbladians}

Before analyzing the general $n$-qubit case, we begin by analyzing candidate degrees of freedom for single-qubit Lindbladians. In order for (\ref{eq:gauge_experiment}) to be equivalent to (\ref{eq:experiment}), we consider conjugation transformations of the Lindbladian: $\mathcal L' = \Lambda \circ \mathcal L \circ \Lambda^{-1}$, as $e^{t\mathcal L'} = \Lambda \circ e^{t\mathcal L} \circ \Lambda^{-1}$. Substituting this into (\ref{eq:gauge_experiment}), we have that if $\Lambda$ commutes with each $\mathcal U_i$, then choosing $\mathcal E_\text{prep}' = \Lambda \circ \mathcal E_\text{prep}$ and $\mathcal E_\text{meas}' = \mathcal E_\text{meas} \circ \Lambda^{-1}$ results in complete equality between (\ref{eq:gauge_experiment}) and (\ref{eq:experiment}). Since we are permitting any unitary transformation, we require $\Lambda$ to commute with all unitary channels. To satisfy this condition, we therefore choose $\Lambda$ to be a depolarizing channel. With this transformation defined, we now prove that this transformation constitutes a valid gauge transformation and identify the components it affects.

\begin{proposition}[Single-Qubit Lindbladian Gauge Dependence]\label{prop:gauge_single}
    There exist SPAM error channels $\mathcal E_\text{prep}, \mathcal E_\text{meas}$ such that for any single-qubit Lindbladian $\mathcal L$, under the gauge transformation
        \begin{align*}
        \mathcal L &\mapsto \mathcal L' = e^{-\theta\Gamma} \circ  \mathcal L  \circ e^{\theta\Gamma} \\
        \mathcal E_\text{prep} &\mapsto \mathcal E_\text{prep}' = e^{-\theta \Gamma} \circ \mathcal E_\text{prep} \\
        \mathcal E_\text{meas} &\mapsto \mathcal E_\text{meas}' = \mathcal E_\text{meas} \circ e^{\theta \Gamma},
    \end{align*}
    where
    $$\Gamma = \frac{1}{4}[3\rho - X \rho X - Y \rho Y - Z \rho Z],$$
    the components $\text{Re}[\chi_{0k}]$ and $\text{Im}[\chi_{ij}]$ are affected by a factor of $e^{-\theta}$ for any $P_i, P_j, P_k \in \{X, Y, Z\}$ with $i \ne j$.
\end{proposition}
To demonstrate the existence of this gauge degree of freedom, we must first show that the transformed error channels are both valid and satisfy the SPAM error bound, after which, we must then demonstrate that the transformed Lindbladian still generates a CPTP map. To validate the error channels, we utilize the following lemma for $n$-qubit depolarizing channels:
\begin{restatable}{lemma}{lemgaugeerrorchannel}\label{lem:gauge_error_channel}
    There exist $n$-qubit depolarizing channels $\mathcal E_1$ and $\mathcal E_2$ satisfying
    $$\| \mathcal E_1 - \mathcal I \|_\diamond + \| \mathcal E_2 - \mathcal I \|_\diamond \le \varepsilon_\text{SPAM}$$
    such that 
    \begin{align*}
        \mathcal E_1' &= e^{-\theta \Gamma} \circ \mathcal E_1 \\
        \mathcal E_2' &= \mathcal E_2 \circ e^{\theta \Gamma}
    \end{align*}
    constitute valid depolarizing channels for $0 \le \theta \le -\ln(1 - \varepsilon_\text{SPAM}/4)$, where
    $$\Gamma = \frac{1}{4^n}\left[(4^n - 1)\rho - \sum_{i \ne 0}P_i\rho P_i\right]$$
    is the generator for $n$-qubit global depolarizing channels and the new error channels satisfy
    $$\| \mathcal E_1' - \mathcal I \|_\diamond + \| \mathcal E_2' - \mathcal I \|_\diamond \le \varepsilon_\text{SPAM}$$
\end{restatable}
Note that this lemma holds for any $n \ge 1$, which will be crucial for the next section, where we deal with the multi-qubit case. This lemma demonstrates that by picking $\mathcal E_\text{prep}, \mathcal E_\text{meas}$ adversarially to be global depolarizing channels, we can generate a collection of error channels that still obey the SPAM constraint for $0 \le \theta \le -\ln(1 - \varepsilon_\text{SPAM}/4)$, indicating that $\theta$ can be of constant order in the worst case. This lemma is rigorously proven in Appendix \ref{chap:gauge}, however, a graphical demonstration is provided in Figure \ref{fig:error_transform}.

Having constrained $\theta$, we can now focus on showing the validity of the Lindbladian. To this end, we utilize the following lemma:
\begin{restatable}{lemma}{lemgaugesinglelindbladian}\label{lem:gauge_single_lindbladian}
    For any single-qubit Lindbladian $\mathcal L$, the transformed operator $\mathcal L' = e^{-\theta \Gamma} \circ \mathcal L \circ e^{\theta \Gamma}$ is a valid Lindbladian for all $\theta \ge 0$.
\end{restatable}
This lemma follows intuitively from the geometric description of single-qubit CPTP maps (see Figure \ref{fig:lindblad_transform}). Specifically, every single-qubit map can be described as a mapping of the Bloch sphere to an ellipsoid contained entirely within \cite{king_minimal_2001}. The parameterization of this ellipsoid is intricately linked to the Pauli transfer matrix of the CPTP map. Most importantly for our purposes, the leftmost column of the Pauli transfer matrix encodes the translation of the center of the ellipsoid relative to the center of the Bloch sphere. Since the gauge transformation simply suppresses this column, we have that our gauge transformation effectively makes the Lindblad evolution more unital. Since the unital version of the Lindblad evolution operator is trivially contained inside the Bloch sphere, the convexity of the Bloch sphere dictates that shrinking the leftmost column by any factor $e^{-\theta}$ for $\theta \ge 0$ will still produce a valid Lindbladian.

\begin{figure}[tbp!]
    \centering
    \includegraphics[width=0.8\linewidth]{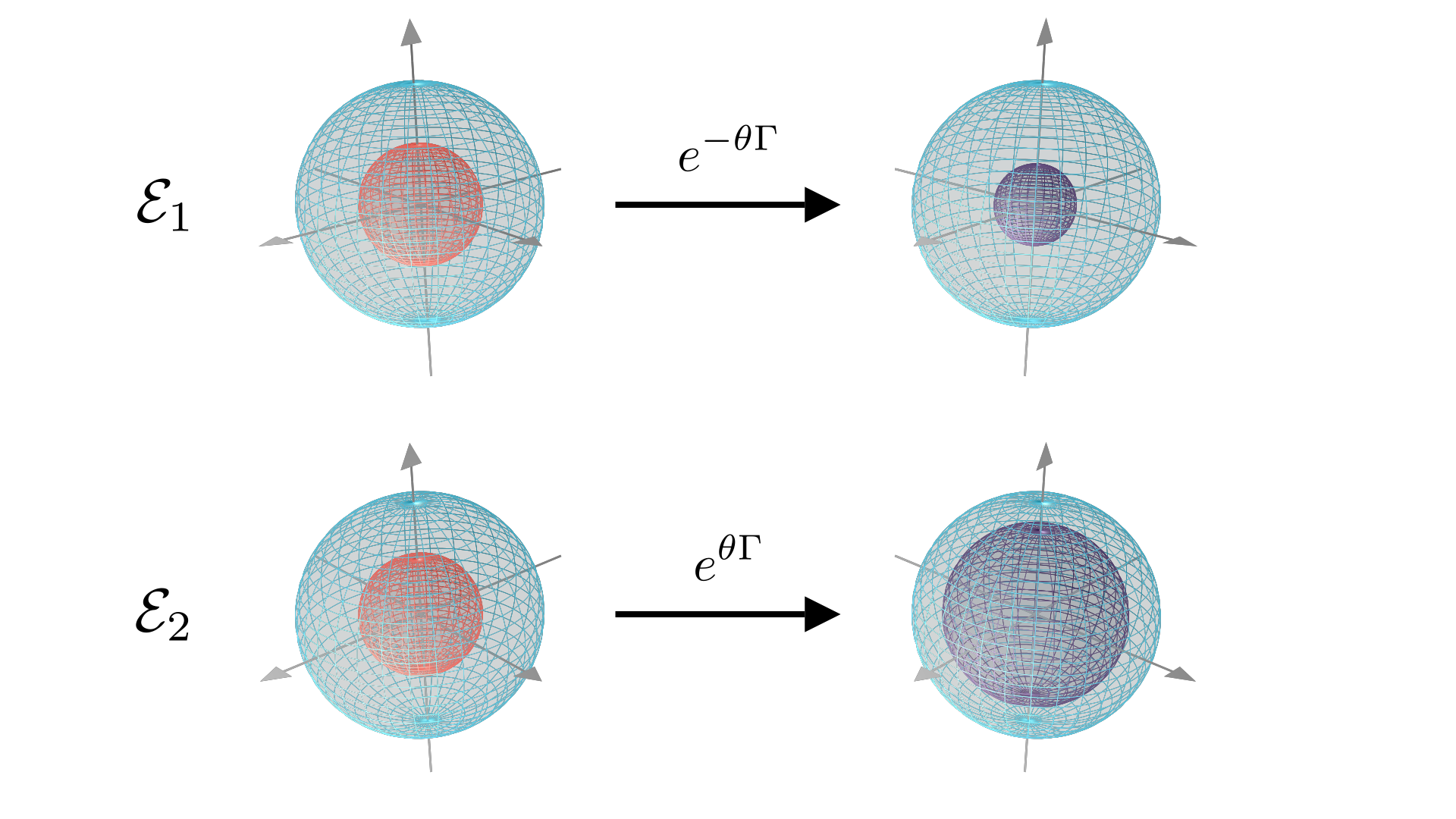}
    \caption{\textbf{Error channels after gauge transformation depicted for a single qubit}. The transformation $e^{-\theta \Gamma}$ acts as a depolarizing channel, compressing the image of $\mathcal E_1$, while $e^{\theta \Gamma}$ expands the image of $\mathcal E_2$. We can only choose $\theta$ such that $\mathcal E_2$ remains a quantum channel (i.e. the map does not expand outside the Bloch sphere) and $\mathcal E_1$ does not move too far from the identity channel $\mathcal I$.}
    \label{fig:error_transform}
\end{figure}

\begin{figure}[tbp]
    \centering
    \includegraphics[width=0.8\linewidth]{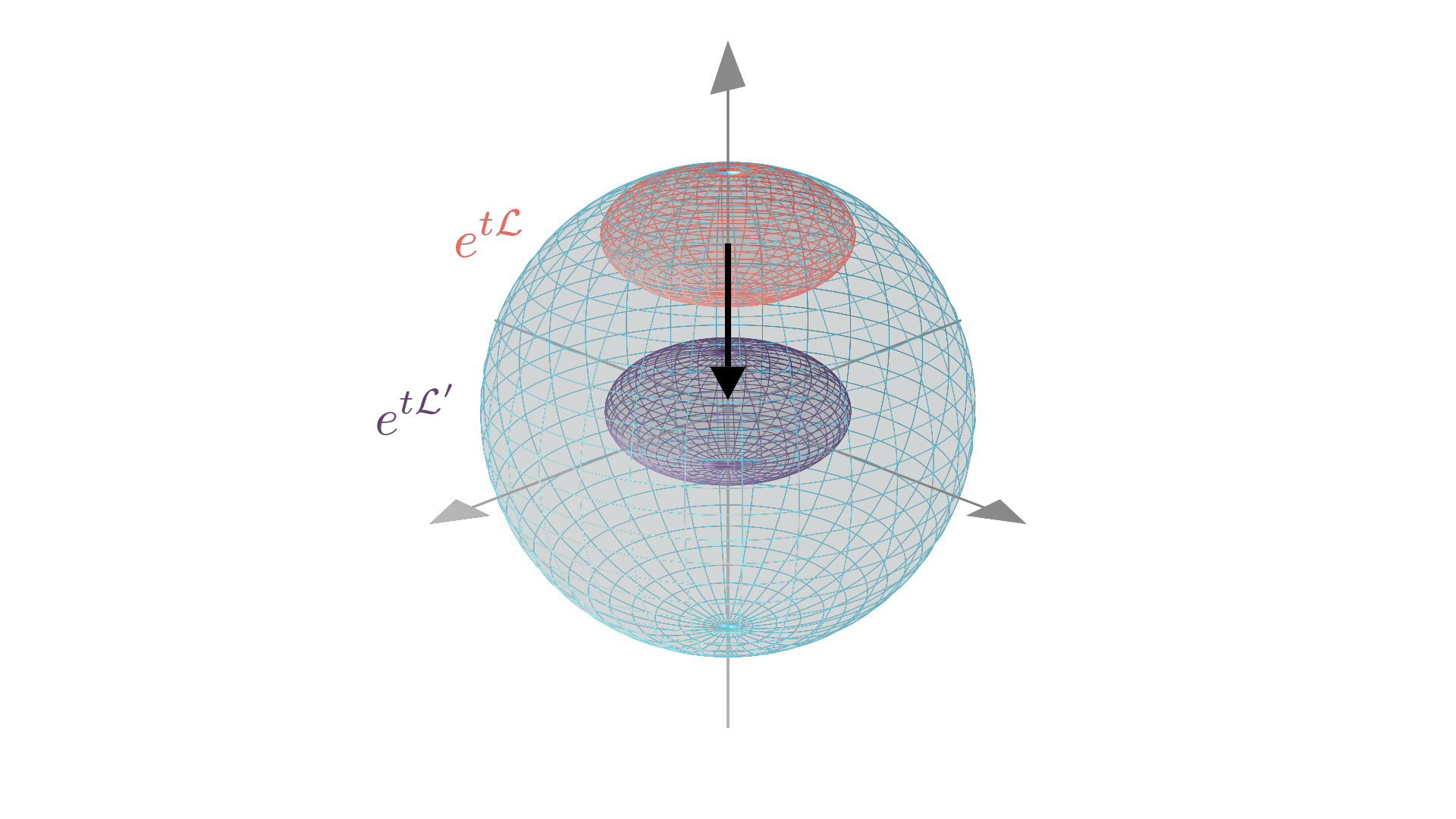}
    \caption{\textbf{Lindblad evolution after gauge transformation depicted for a single qubit}. The transformation $e^{-\theta \Gamma} \cdot e^{\theta \Gamma}$ shrinks the deviation of the center of the ellipsoid, shifting it towards the origin, effectively making the time evolution operator more unital. Since this ellipsoid always remains within the Bloch sphere for $\theta \ge 0$, we have that $\mathcal L'$ is always a valid Lindbladian in the single-qubit case.}
    \label{fig:lindblad_transform}
\end{figure}

To conclude the proof for Proposition \ref{prop:gauge_single}, we now analyze the largest possible shift that can occur in the coefficient values. From Lemma \ref{lem:gauge_error_channel}, we had that $0 \le \theta \le -\ln(1 - \varepsilon_\text{SPAM}/4),$
assuming the prescribed form of the noise. Consequently, we have that the multiplicative shift is given as
$$1 - e^{-\theta} = \varepsilon_\text{SPAM}/4$$
As mentioned previously, the only affected component of the Lindbladian was the leftmost column of the PTM. The elements in this column correspond directly to linear combinations of $\text{Im}[\chi_{ij}]$ and $\text{Re}[\chi_{0k}]$ for $P_i, P_j, P_k \in \{X, Y, Z\}$, with $i \ne j$. Since these components $\chi_{ab}$ can be changed by a multiplicative factor of $\mathcal O(\varepsilon_\text{SPAM})$, we have that if $\varepsilon_\text{SPAM} \sim O(1)$ as is typical, then we cannot learn these components to $\epsilon$-accuracy, thereby demonstrating Proposition \ref{prop:gauge_single}.

\subsection{Multi-Qubit Lindbladians}

Having established this degree of freedom in the single-qubit case, we are now ready to extend it to the $n$-qubit regime. As we increase the dimensionality of the system, we observe that every component of the leftmost column of the Pauli transfer matrix is composed of $4^n$ coefficients of the Chi matrix. Consequently, changing any of these components requires updating $4^n$ components. Naturally, this means that unlike the single-qubit case, the global depolarizing channel approach will not necessarily yield a valid Lindbladian when applied to every $n$-qubit Lindbladian. As a simple counterexample, consider the Lindbladian
$$\mathcal L[\rho] = \frac{1}{4}(XI + iYI)\rho(XI - iYI) + \frac{1}{4}\{II - ZI, \rho\}$$
This is a single-qubit Lindbladian extended to act trivially on another qubit. The only nontrivial element in the PTM here corresponds to $|ZI\rangle\rangle\langle\langle II|$ (here we denote $|P\rangle\rangle$ as the vector corresponding to Pauli $P$). Expanding this out in terms of an operator expansion, we have that
$$|ZI\rangle\rangle\langle\langle II| = ZI\sum_{i}P_i \rho P_i$$
Hence, we must add 16 Pauli terms to the Lindbladian to alter this value. Since the only Paulis in the original Lindbladian are $II, XI, YI, ZI$, which are closed under multiplication (up to a phase), this means that we would need to add Paulis that cannot be formed from products of Paulis in this set. Hence, the resulting superoperator cannot be a Lindbladian as the Chi matrix would no longer be PSD.

Therefore, rather than showing that there exists a nontrivial gauge transformation under which all Lindbladians transform into Lindbladians, we instead shift our focus to constructing families of Lindbladians that cannot be distinguished in the presence of SPAM error. In Appendix \ref{chap:gauge}, we demonstrate that by constructing Lindbladians consisting of a global depolarizing generator and depolarizing channels that have been left multiplied by specific Paulis, we can construct Lindbladians that are gauge-dependent, satisfying the following lemma:
\begin{restatable}{lemma}{lemgaugemultilindbladian}\label{lem:gauge_multi_lindbladian}
    There exists a Lindbladian $\mathcal L[\rho] = \sum_{a, b}\chi_{ab}P_a \rho P_b$ such that under the depolarizing transformation
    $$\mathcal L'(\theta) = e^{-\theta \Gamma} \circ \mathcal L \circ e^{\theta \Gamma},$$
    the resulting superoperator $\mathcal L'  = \sum_{a, b}\chi_{ab}'P_a \rho P_b$ is a Lindbladian for $\theta \ge 0$, and $\text{Im}[\chi_{ab}']$ differs from $\text{Im}[\chi_{ab}]$ by a multiplicative factor of $e^{-\theta}$ for $\{P_a, P_b\} = 0$ and $\text{Re}[\chi_{ab}']$ differs from $\text{Re}[\chi_{ab}]$ by a multiplicative factor of $e^{-\theta}$ for $[P_a, P_b] = 0$.
\end{restatable}
Taken together, Lemmas \ref{lem:gauge_error_channel} and \ref{lem:gauge_multi_lindbladian} give us the following theorem on unlearnability.
\begin{theorem}[Unlearnability]\label{thm:unlearnability}
    There exist SPAM error channels $\mathcal E_1, \mathcal E_2$ satisfying the threshold $\varepsilon_\text{SPAM}$ such that the following holds true:

    For every Pauli pair $P_i, P_j$, there exists a Lindbladian $\mathcal L = \sum_{i, j}\chi_{ij}P_i\rho P_j$ for which components $\text{Im}[\chi_{ij}]$ or $\text{Re}[\chi_{ij}]$ cannot be learned to arbitrary accuracy, depending on whether $P_i, P_j$ anticommute or commute, respectively.
\end{theorem}

\section{Learnable Coefficients}

Having established the gauge-dependent components of the Lindbladian, we now demonstrate how to learn the remaining components. The protocol is described in full in Figure \ref{fig:lindblad_learn_spam}.

\begin{figure}[tbp]
    \centering
    \includegraphics[width=0.8\linewidth]{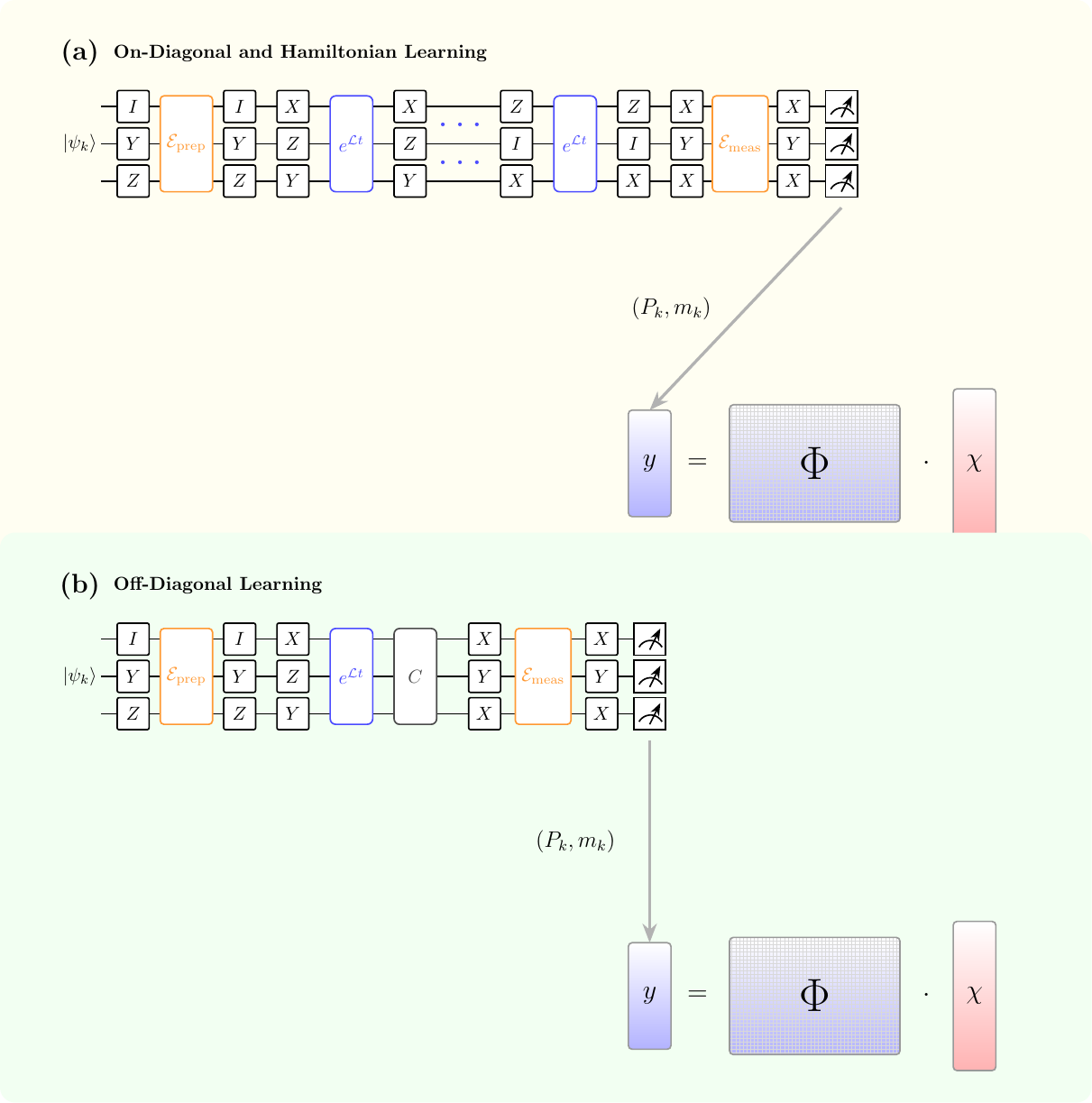}
    \caption{\textbf{SPAM-Robust Lindbladian Learning.} (a) Protocol for learning diagonal and Hamiltonian components of Lindbladian $\mathcal L = \sum_{i, j}\chi_{ij}P_i \rho P_j$ using eigenvalue $m_k$ from Pauli twirled evolution on a $P_k$ $+1$-eigenstate and compressed sensing. (b) Modified protocol for estimating off-diagonals by inserting a Clifford $C$ after Lindblad evolution.}
    \label{fig:lindblad_learn_spam}
\end{figure}

\subsection{Diagonals}

Just as before, we begin with the on-diagonal components. Without Bell sampling, we no longer possess the ability to test the support of the Lindbladian. Consequently, we must fall back on estimation procedures that utilize linear systems to recover the Lindbladian coefficients from strategic measurements. To see how this manifests for the diagonal components, we consider the simplified case where all the nonzero coefficients lie on the diagonal, meaning we have a Pauli Lindbladian: 
\begin{lemma}[SPAM-Robust Learning of Pauli Lindbladians]
    \label{lem:spam_pauli_channel}
    Given a Pauli Lindbladian $\mathcal L(\rho) = \sum_{P} \gamma_P[P \rho P - \rho]$, we can learn the coefficients $\gamma_P$ in SPAM-robust fashion.
\end{lemma}

\begin{proof}
    Since all the jump operators are Paulis, we can express any Pauli Lindbladian in the following form:
    $$\mathcal L[\rho] = \sum_P \gamma_P[P \rho P - \rho]$$
    In the Heisenberg picture, we therefore have that
    $$\mathcal L^*[Q] = \sum_P \gamma_P[P Q P - Q] = \sum_P \gamma_P[((-1)^{\langle P, Q\rangle} - 1)Q] = -2\sum_P \gamma_P\langle P, Q\rangle Q,$$
    where
    $$\langle P, Q\rangle  = \begin{cases}
        1 & \{P, Q\} = 0 \\
        0 & [P, Q] = 0
    \end{cases}$$
    Hence, we have that any $Q \in \mathcal P_n$ is an eigenoperator of $\mathcal L^\dag$. As such, we have that
    $$e^{t\mathcal L^*}[Q] = e^{-2t\sum_P \gamma_P\langle P, Q\rangle}Q$$
    Consequently, by choosing a $Q$ eigenstate with eigenvalue $+1$, we can compute the rates $\gamma_P$ by solving the resulting system. Note that the contributions above are from Paulis $P$ which anticommute with $Q$.

    We now account for the effects of SPAM error on the system. Suppose we have state preparation and measurement error channels $\mathcal E_1, \mathcal E_2$. If we now apply Pauli twirling to each of these channels, we get effective Pauli channels $\mathcal E_1', \mathcal E_2'$. Let us denote the eigenvalues corresponding to each Pauli $Q$ as $\lambda_{1, Q}$, $\lambda_{2, Q}$. If we now apply these noisy channels in tandem with the Lindbladian evolution, we have that
    $$(\mathcal E_2' \circ e^{t\mathcal L^*} \circ \mathcal E_1')[Q] = \lambda_{1, Q}\lambda_{2, Q}e^{-2t\sum_P \gamma_P\langle P, Q\rangle}Q$$
    We can now evolve for two different times $t_1, t_2$ to learn the eigenvalues corresponding to $Q$ and then divide them to remove the effect of the SPAM error. With enough choices of Paulis $Q$, we can then estimate all the coefficients since each linear combination of $\gamma_P$ coefficients is linearly independent.
\end{proof}
As we saw, by estimating ratios of Pauli eigenvalues to a certain accuracy threshold, we can compute the coefficients $\gamma_P$ by solving the resulting linear system. We can transform any Lindbladian into a Pauli Lindbladian through the application of Pauli twirling. We define the twirled Lindbladian as follows:
\begin{equation}
    \mathcal L_\text{tw} = \mathbb E_{P \sim \mathcal P_n}[P\mathcal L[P\rho P]P]
\end{equation}
In order to approximate evolution under $\mathcal L_\text{tw}$, we can evolve under $\mathcal L$ for short time and then twirl, repeating this procedure until we evolve for the full duration $t$. To bound the error incurred in this procedure, we utilize the following lemma:
\begin{restatable}[Lindbladian Pauli Twirling]{lemma}{lemlindbladpaulitwirl}\label{lem:lindblad_pauli_twirl}
    Given an $M$-sparse Lindbladian $\mathcal L$, let $\mathcal E_t[\rho] = e^{t\mathcal L_\text{tw}}$ be the evolution operator under the twirled Lindbladian. Define $\mathcal E'_{t, r}[\rho] = (\mathbb E_{P \sim \mathcal P_n}[Pe^{t\mathcal L/r}[P\rho P]P])^r$. We can then bound the diamond distance between these channels as follows:
    $$\|\mathcal E'_{t, r} - \mathcal E_{t}\|_\diamond \le \mathcal O\left(\frac{M^3t^2}{r}\right),$$
    for $t \le \mathcal O(r/M^{2})$.
\end{restatable}
We derive this lemma in Appendix \ref{chap:spam_lindblad_proofs} by considering a telescoping sum argument, where we then use the number of terms in the Lindbladian to bound the diamond distance. Using the Hölder Matrix Inequality \cite{baumgartner_inequality_2011}, we have that for any observable $O$,
$$\text{Tr}[O(\mathcal E'_{t, r}[\rho] - \mathcal E_t[\rho])] \le \|O\|_p\|\mathcal E'_{t, r}[\rho] - \mathcal E_t[\rho]\|_q,$$
where $1/p + 1/q = 1$ for $1 \le p, q \le \infty$. Choosing $p = \infty$ and $q = 1$, we have that $\|O\|_\infty = 1$ for any Pauli observable $O$, meaning that
$$\text{Tr}[O(\mathcal E[\rho] - \mathcal E'[\rho])] \le \|\mathcal E[\rho] - \mathcal E'[\rho]\|_1 \le \|\mathcal E - \mathcal E'\|_\diamond$$

Consequently, since we bounded the diamond distance before, we have that choosing sufficiently large $r$ allows us to control the error in the eigenvalue estimation. Nevertheless, we must first figure out the accuracy to which we must learn the eigenvalue in the first place. We can do this with the following lemma:

\begin{restatable}{lemma}{lemeigenuncertainty}\label{lem:eigen_uncertainty}
    Given $f_1 = Ae^{\gamma t}, f_2 = Ae^{2\gamma t}$, we have that to estimate $\gamma$ to accuracy $\mathcal O(\eta)$, we need to estimate $f_1, f_2$ to accuracy $\mathcal O(A\eta t e^{\gamma t}), \mathcal O(A\eta t e^{2\gamma t})$, respectively.
\end{restatable}
The uncertainty analysis for this result is presented in Appendix \ref{chap:spam_lindblad_proofs}. Dedicated error analysis of the linear system in Appendix \ref{chap:linear} demonstrates that we need to estimate $-2\sum_P \gamma_P \langle P, Q \rangle$ to accuracy $\mathcal O(\epsilon)$. Taken together, this means that if we choose the maximum evolution time $t \sim \mathcal O(1/M)$ (since $\gamma \sim \mathcal O(M)$ in the worst case), we need to learn the eigenvalues at the two times to accuracy $\mathcal O(A\epsilon/M)$, where $A = \lambda_{1, Q}\lambda_{2, Q}$. Consequently, the smaller the eigenvalue is, the more accurately we have to estimate this coefficient, as expected. To quantify how small these eigenvalues can possibly be, we use the SPAM error bound. To this end, we have the following lemma:
\begin{restatable}{lemma}{lemspameigenvaluebound}\label{lem:spam_eigenvalue_bound}
    Given the SPAM error bound $\varepsilon_\text{SPAM}$ as stated in Definition \ref{def:spam}, the eigenvalues for the twirled state preparation and measurement error channels, notated $\mathcal E_1', \mathcal E_2'$, respectively, the absolute value of the product of their eigenvalues is lower bounded by
    $$|\lambda_{1, Q}\lambda_{2, Q}| \ge 1 - \varepsilon_\text{SPAM}$$
    for $\varepsilon_\text{SPAM} \le 1$.
\end{restatable}
Consequently, performing the same Hoeffding analysis as before, we choose a sample complexity of $N = \widetilde{\mathcal O}(M^2/\epsilon^2)$ to estimate these eigenvalues to the desired accuracy. Using the established compressed sensing result \cite{haviv_restricted_2017} stated in Theorem \ref{thm:sub_sample}, we have that $K = \mathcal O(Mn\log^2 M)$ probes are sufficient to learn these coefficients, giving us the following proposition.

\begin{restatable}[Complexity of SPAM-Robust On-Diagonal Coefficient Learning]{proposition}{propspamondiagonal}\label{prop:spam_on_diagonal}
    Given an $M$-sparse Lindbladian $\mathcal L$ with dissipator support $\hat {\mathcal S}$, we can learn the diagonal components of $\chi(\mathcal L)$ to $\epsilon$-accuracy in SPAM-robust fashion in total evolution time
    $$T = {\mathcal O}\left(\frac{M^2n\log^2 M}{\epsilon^2}\right)$$
    with high probability.
\end{restatable}

Furthermore, the algorithm for Proposition \ref{prop:spam_on_diagonal} is given as follows.
\begin{algorithm}[H]
\caption{SPAM-Robust On-Diagonal Lindbladian Coefficient Learning}\label{alg:spam_on_diagonal}
\begin{algorithmic}[1]
\REQUIRE Query access to $\mathcal E_t[\cdot] = e^{t \mathcal L}[\cdot]$, accuracy threshold $\epsilon$, sparsity bound $M$
\STATE Initialize all $\hat \chi_{kk} = 0$.
\STATE Set $t \leftarrow \mathcal O(1/M)$
\STATE Set $r \leftarrow \mathcal O(M^2/\epsilon)$
\STATE Set $K \leftarrow \mathcal O(Mn\log^2 M)$
\STATE Set $N \leftarrow \widetilde {\mathcal O}\left(M^2/\epsilon^2\right)$ 
\FOR{$k = 1, \dots, K$}
\STATE Set $p_k^{(1)}, p_k^{(2)} \leftarrow 0$
\STATE Sample Pauli $Q_k \in \mathcal P_n$
    \FOR{$\ell = 1, \dots, N$}
        \STATE Prepare state $\rho_k = \frac{1 + Q_k}{2}$, twirling before and after state preparation.
        \STATE Evolve under $\mathcal E_{t, r}' = (\mathbb E_{P \sim \mathcal P_n}[P\mathcal E_{t/r}[P \rho P]P])^r$
        \STATE Measure $Q_k$ to obtain $m_k$, twirling before and after, updating $p_k^{(1)} \leftarrow p_k^{(1)} + m_k/N$
    \ENDFOR
    \FOR{$\ell = 1, \dots, N$}
        \STATE Prepare state $\rho_k = \frac{1 + Q_k}{2}$, twirling before and after state preparation.
        \STATE Evolve under $\mathcal E_{2t, r}'$
        \STATE Measure $Q_k$ to obtain $m_k$, twirling before and after, updating $p_k^{(2)} \leftarrow p_k^{(2)} + m_k/N$
    \ENDFOR
\ENDFOR
\STATE Let $\gamma_k \leftarrow \frac{1}{t}\ln\left(p_k^{(2)}/p_k^{(1)}\right)$ for all $k$
\STATE Solve compressed sensing problem given $\{Q_k, \gamma_k\}_{k \in [K]}$ using $\ell_1$-optimization to recover $\hat \chi_{kk}$ and support $\hat {\mathcal S}$
\RETURN $\hat \chi_{kk}, \hat {\mathcal S}$
\ENSURE Nonzero learned coefficients $\hat \chi_{kk}$ and dissipator support $\hat {\mathcal S}$
\end{algorithmic}
\end{algorithm}

\subsection{Off-Diagonals}

Having demonstrated how to learn the on-diagonals, we now address the off-diagonals. One potential approach is to conjugate the Lindbladian by $T$ and Clifford gates and apply the same twirling procedure as before, allowing us to evolve for a longer time whilst enabling the off-diagonal coefficients to show up on the diagonals. However, this approach only works for $\text{Re}[\chi_{ij}]$ where $\{P_i, P_j\} = 0$. Intuitively, this is because there can be no unitary transformation such that $P_i \mapsto \frac{P_i \pm iP_j}{2}$ for commuting $P_i, P_j$ since unitary transformations preserve eigenvalues.

Consequently, we instead refer to the approach we utilized in the noiseless case, where we simply applied linear combinations of Paulis before performing Bell measurements to extract the off-diagonal components. Since the Pauli-Lindblad evolution we were performing previously was twirled to remove all the off-diagonal contributions, we now need to evolve without suppressing these contributions. However, this also means that we cannot evolve for as much time since the higher-order terms will become significant. To this end, we consider the evolution
\begin{equation}\label{eq:spam_off_diag}
    \mathbb E_{P \sim \mathcal P_n}[P (\mathcal C_{ij}\circ e^{\mathcal Lt})[\rho] P]
\end{equation}
to learn either the real or imaginary component of the coefficient $\chi_{ij}$. Here, we have that
$$(\mathcal C_{ij}\circ e^{\mathcal Lt})[\rho] \approx \mathcal C_{ij}[\rho + t\mathcal L[\rho] + \dots] \approx \frac{1}{2}P_i \rho P_i + \frac{1}{2}P_j \rho P_j + t\sum_{k}\chi_{kk}'P_k \rho P_k + (\text{non-diag})$$
We denote the eigenvalue corresponding to this effective channel as $\lambda_Q$. Since the error decays in the same fashion as it did for the noiseless leaning protocol, we again choose $t = \mathcal O(\epsilon/M^2)$. In Appendix \ref{chap:spam_lindblad_proofs}, we demonstrate that this choice of $t$ is not only sufficient to suppress the Taylor series error, but also to measure the eigenvalues themselves to the desired accuracy. However, it should be noted that for this procedure, rather than evolving from times $t$ and $2t$, we instead opt for times $0$ and $t$ to measure the SPAM contribution directly. This is because the eigenvalue can be close to zero for this procedure, making it difficult to estimate the square of the Pauli channel eigenvalue, like we were able to do last time. Using the same Hoeffding analysis as before as well as $K = \mathcal O(Mn\log^2 M)$ Pauli probes for each of the $\mathcal O(M^2)$ off-diagonals, this gives us the following proposition:
\begin{restatable}[Complexity of SPAM-Robust Off-Diagonal Coefficient Learning]{proposition}{propspamoffdiagonal}\label{prop:spam_off_diagonal}
    Given an $M$-sparse Lindbladian $\mathcal L$ with dissipator support $\hat {\mathcal S}$, we can learn the off-diagonal components of $\chi(\mathcal L)$ to $\epsilon$-accuracy in SPAM-robust fashion in total evolution time
    $$T = \mathcal O\left(\frac{M^5n \log^2 M}{\epsilon^3}\right)$$
    with high probability.
\end{restatable}

The corresponding algorithm for this proposition is given as follows:
\begin{algorithm}[H]
\caption{SPAM-Robust Off-Diagonal Lindbladian Coefficient Learning}\label{alg:spam_off_diagonal}
\begin{algorithmic}[1]
\REQUIRE Query access to $\mathcal E_t[\cdot] = e^{t \mathcal L}[\cdot]$, accuracy threshold $\epsilon$, sparsity bound $M$, learned dissipator support $\hat {\mathcal S}$
\STATE Initialize all $\hat \chi_{kk} = 0$.
\STATE Set $t \leftarrow \mathcal O(\epsilon/M^2)$
\STATE Set $K \leftarrow \mathcal O(Mn\log^2 M)$
\STATE Set $N \leftarrow \widetilde {\mathcal O}\left(M^4/\epsilon^4\right)$ 
\STATE Define $\hat {\mathcal S}^{(2)} = \{a \oplus b: a,b \in \hat{\mathcal S}\}$
\STATE Define coefficient support $\hat{ \mathcal S}' = (\mathcal S \times \mathcal S) \cup (\hat {\mathcal S}^{(2)} \times \{0\}) \cup (\{0\} \times \hat {\mathcal S}^{(2)})$
\FOR{$(a, b) \in \hat{ \mathcal S}'$ where $i < j$}
    \FOR{$k = 1, \dots, K$}
    \STATE Set $p_k^{(0)}, p_k^{(1)} \leftarrow 0$
    \STATE Sample Pauli $Q_k \in \mathcal P_n$
        \FOR{$\ell = 1, \dots, N$}
            \STATE Prepare state $\rho_k = \frac{1 + Q_k}{2}$, twirling before and after state preparation.
            \STATE Evolve under $\mathcal C_{ab} \circ \mathcal E_{t}$
            \STATE Measure $Q_k$ to obtain $m_k$, twirling before and after, updating $p_k^{(1)} \leftarrow p_k^{(1)} + m_k/N$
        \ENDFOR
        \FOR{$\ell = 1, \dots, N$}
            \STATE Prepare state $\rho_k = \frac{1 + Q_k}{2}$, twirling before and after state preparation.
            \STATE Measure $Q_k$ to obtain $m_k$, twirling before and after, updating $p_k^{(0)} \leftarrow p_k^{(0)} + m_k/N$
        \ENDFOR
    \ENDFOR
    \STATE Let $\gamma_k \leftarrow \frac{1}{t}\ln\left(p_k^{(1)}/p_k^{(0)}\right)$ for all $k$
    \STATE Solve compressed sensing problem given $\{Q_k, \gamma_k\}_{k \in [K]}$ using $\ell_1$-optimization to recover $\hat \chi_{jj}'$
    \STATE Compute $\hat \chi_{ab} = \frac{i^{\langle a, b \rangle}}{2}(\hat \chi_{kk}' - \hat \chi_{00}')$ for $k = a \oplus b$.
\ENDFOR
\RETURN $\hat \chi_{ab}$
\ENSURE Nonzero learned gauge-independent off-diagonal coefficients $\hat \chi_{ab}$
\end{algorithmic}
\end{algorithm}

\subsection{Hamiltonian Coefficients}

The last part of the Lindbladian that we can always estimate without worrying about gauge-degrees of freedom are the Hamiltonian terms. Just as before, we consider the case where the Hamiltonian terms are outside $\{a \oplus b : a, b \in \mathcal S\}$. To do this, we need to consider a modification to Lemma \ref{lem:lindblad_pauli_twirl}:

\begin{restatable}[Second-Order Lindbladian Pauli Twirling]{lemma}{lemlindbladpaulitwirlquad}\label{lem:lindblad_pauli_twirl_quad}
    Given an $M$-sparse Lindbladian $\mathcal L$, let $\mathcal E_{t, r}[\rho] = e^{t\mathcal L_\text{tw} + \frac{t^2}{2r}[(\mathcal L^2)_\text{tw} - (\mathcal L_\text{tw})^2]}$. Define $$\mathcal E'_{t, r}[\rho] = (\mathbb E_{P \sim \mathcal P_n}[Pe^{t\mathcal L/r}[P\rho P]P])^r$$
    We can then bound the diamond distance between these channels as follows:
    $$\|\mathcal E'_{t, r} - \mathcal E_{t, r}\|_\diamond \le \mathcal O\left(\frac{M^5t^3}{r^2}\right),$$
    for $t \le \mathcal O(r/M^2)$
\end{restatable}
Hence, we have that the trotterized evolution that we performed previous can also be used to approximate evolution under $\mathcal L_\text{tw} + \frac{t}{2r}[(\mathcal L^2)_\text{tw} - (\mathcal L_\text{tw})^2]$ for time $t$. Since everything is twirled here, we have that this is still diagonal in the Pauli basis. Writing the Hamiltonian as $H = \sum_j h_j P_j$, we have that for the Hamiltonian terms outside the support of the dissipator, the coefficient of the term $P_j \rho P_j$ is given as $\frac{h_p^2t}{r}$. Hence, to estimate $h_p$ to $\mathcal O(\epsilon)$-accuracy, we need to estimate this coefficient to $\mathcal O(\epsilon^2t/r)$-accuracy. Consequently, we have that the compressed sensing problem now involves an $\mathcal O(M^2)$-sparse solution vector and coefficient estimation to $\mathcal O(\epsilon^2t/r)$-accuracy. Performing the same error analysis as before, we now find that we need to learn each eigenvalue to accuracy $\mathcal O(\epsilon^4/M^6)$, which by the same Hoeffding derivation gives us the following result:
\begin{restatable}[Complexity of SPAM-Robust Hamiltonian Coefficient Learning]{proposition}{propspamhamiltonian}\label{prop:spam_hamiltonian}
    Given an $M$-sparse Lindbladian $\mathcal L$ with dissipator support $\hat {\mathcal S}$, we can learn the Hamiltonian components of $\chi(\mathcal L)$ to $\epsilon$-accuracy in SPAM-robust fashion in total evolution time
    $$T = {\mathcal O}\left(\frac{M^{13}n\log^2 M}{\epsilon^8}\right)$$
    with high probability.
\end{restatable}
As we can see, this quadratic order terms are much more difficult to estimate due to their suppression under twirling. Consequently, the runtime of the algorithm will be bottlenecked by the Hamiltonian learning protocol.

The algorithm is nearly identical to the on-diagonal learning one, with the only difference being different parameter choices:
\begin{algorithm}[H]
\caption{SPAM-Robust Hamiltonian Coefficient Learning}\label{alg:spam_hamiltonian}
\begin{algorithmic}[1]
\REQUIRE Query access to $\mathcal E_t[\cdot] = e^{t \mathcal L}[\cdot]$, accuracy threshold $\epsilon$, sparsity bound $M$, learned dissipator support $\hat {\mathcal S}$
\STATE Initialize all $\hat \chi_{ij} = 0$.
\STATE Set $t \leftarrow \mathcal O(1/M)$
\STATE Set $r \leftarrow \mathcal O(M^4/\epsilon^2)$
\STATE Set $K \leftarrow \mathcal O(M^2n\log^2 M)$
\STATE Set $N \leftarrow \widetilde {\mathcal O}\left(M^{12}/\epsilon^8\right)$ 
\FOR{$k = 1, \dots, K$}
\STATE Set $p_k^{(1)}, p_k^{(2)} \leftarrow 0$
\STATE Sample Pauli $Q_k \in \mathcal P_n$
    \FOR{$\ell = 1, \dots, N$}
        \STATE Prepare state $\rho_k = \frac{1 + Q_k}{2}$, twirling before and after state preparation.
        \STATE Evolve under $\mathcal E_{t, r}' = (\mathbb E_{P \sim \mathcal P_n}[P\mathcal E_{t/r}[P \rho P]P])^r$
        \STATE Measure $Q_k$ to obtain $m_k$, twirling before and after, updating $p_k^{(1)} \leftarrow p_k^{(1)} + m_k/N$
    \ENDFOR
    \FOR{$\ell = 1, \dots, N$}
        \STATE Prepare state $\rho_k = \frac{1 + Q_k}{2}$, twirling before and after state preparation.
        \STATE Evolve under $\mathcal E_{2t, r}'$
        \STATE Measure $Q_k$ to obtain $m_k$, twirling before and after, updating $p_k^{(2)} \leftarrow p_k^{(2)} + m_k/N$
    \ENDFOR
\ENDFOR
\STATE Let $\gamma_k \leftarrow \frac{1}{t}\ln\left(p_k^{(2)}/p_k^{(1)}\right)$ for all $k$
\STATE Solve compressed sensing problem given $\{Q_k, \gamma_k\}_{k \in [K]}$ using $\ell_1$-optimization to recover all $\hat \chi_{ij}$
\RETURN $\hat \chi_{ij}$
\ENSURE Nonzero learned coefficients $\hat \chi_{ij}$, including the Hamiltonian terms
\end{algorithmic}
\end{algorithm}

\section{Full SPAM-Robust Learning Algorithm}

Combining the three previous algorithms, we obtain the full protocol for learning Lindbladians in SPAM-robust fashion:
\begin{algorithm}[H]
\caption{SPAM-Robust Sparse Lindbladian Learning}\label{alg:spam_robust_lindblad_learn}
\begin{algorithmic}[1]
\REQUIRE Query access to $\mathcal E_t[\cdot] = e^{t \mathcal L}[\cdot]$, accuracy threshold $\epsilon$, sparsity bound $M$
\STATE Run Algorithm \ref{alg:spam_on_diagonal} to recover the on-diagonals $\hat \chi_{kk}$ and dissipator support set $\hat{\mathcal S}$
\STATE Run Algorithm \ref{alg:spam_off_diagonal} to learn off-diagonals $\hat \chi_{ab}$
\STATE Run Algorithm \ref{alg:spam_hamiltonian} to estimate off-diagonals $\hat \chi_{ab}$ that have not yet been learned
\RETURN $(\hat \chi_{ab})$
\ENSURE Learned Chi matrix $(\hat \chi_{ab})$ for Lindbladian $\mathcal L$
\end{algorithmic}
\end{algorithm}

Furthermore, we have that the runtime is given by the following theorem:
\begin{theorem}
    Given an unknown $M$-sparse, $n$-qubit Lindbladian $\mathcal L[\rho] = \sum_{i, j} \chi_{ij}P_i\rho P_j$, we can learn the gauge-independent components of the chi matrix $\hat \chi_{ij}$ in the presence of SPAM-error $\varepsilon_\text{SPAM} \le 1 - \mathcal O(1)$ in evolution time
    $$T = \mathcal O\left(\frac{M^{13}n \log^2 M}{\epsilon^8}\right)$$
    such that for all gauge-independent $\chi_{ij}$,
    $$|\hat \chi_{ij} - \chi_{ij}| \le \epsilon,$$
    with high probability.
\end{theorem}
Consequently, we can indeed estimate all the gauge-independent components of the Lindbladian using evolution time $T = \text{poly}(M, n, 1/\epsilon)$. Nevertheless, one important caveat of this technique is that solving the resulting linear system requires nontrivial classical postprocessing in general. In particular, to solve the $\ell_1$ optimization problem used to compute the compressed sensing solution via a technique like LARS \cite {efron_least_2004}, we require a computational cost of $\mathcal O(KN^2) = \widetilde{\mathcal O}(16^nMn)$, which is exponential in the number of qubits.

To get around this issue, we can also consider the case where the support of both the Lindbladian and Hamiltonian are known a priori, which is the case under certain reasonable physical assumptions such as $k$-locality. Firstly, this means that the separate Hamiltonian term measuring procedure is no longer required, as these terms can simply be measured during the typical cross-term analysis. Furthermore, from a computational perspective, knowing the support transforms the $\ell_1$ optimization problem into a more tractable $\ell_2$ optimization problem, which gives us polynomial runtime, as analyzed in Appendix \ref{chap:linear}. This then gives us the following theorem:
\begin{theorem}
    Given an unknown $n$-qubit Lindbladian $\mathcal L[\rho] = \sum_{i, j} \chi_{ij}P_i\rho P_j$ with known support superset $\mathcal S$ for the dissipator and $\mathcal S_H$ for the Hamiltonian, where $|\mathcal S \cup \mathcal S_H| \le M$ we can learn the gauge-independent components of the chi matrix $\hat \chi_{ij}$ in the presence of SPAM error $\varepsilon_\text{SPAM} \le 1 - \mathcal O(1)$ in evolution time
    $$T = \mathcal O\left(\frac{M^5n \log^2 M}{\epsilon^3}\right)$$
    and classical post-processing time $\mathcal O(M^5n \log^2 M)$ such that for all gauge-independent $\chi_{ij}$,
    $$|\hat \chi_{ij} - \chi_{ij}| \le \epsilon,$$
    with high probability.
\end{theorem}
This then gives us a similar protocol to Algorithm \ref{alg:spam_robust_lindblad_learn}, where instead of using $\ell_1$ optimization, we can use $\ell_2$ optimization constraining the measurement matrix based on the support sets given. Furthermore, the Hamiltonian learning protocol is no longer necessary, since we can instead learn all the off diagonals using $\hat{\mathcal S}_H$ and Algorithm \ref{alg:spam_off_diagonal}. For detailed analysis of the time complexity of solving the linear system with knowledge of the support, see Appendix \ref{chap:linear}.

\chapter{Conclusions and Outlook}

In this thesis, we developed the first efficient ansatz-free Lindbladian learning algorithm with polynomial-time classical postprocessing. We achieved this by using Bell sampling to probe the support of both the dissipative and Hamiltonian parts of the Lindbladian and then measure the corresponding coefficients based on the probability of sampling each outcome. Furthermore, we provide the first characterization of Lindbladian gauge dependence in the presence of unstructured SPAM error as well as the first efficient SPAM-robust procedure for learning gauge-independent components.

From an experimental standpoint, Lindbladian learning offers a method of characterizing Markovian noise in quantum devices \cite{severin_learning_2026}. Unlike data-intensive techniques such as quantum process tomography, this algorithm facilitates sample- and time-efficient reconstruction of sparse, nonlocal error models without assuming structural knowledge \textit{a priori}. Nevertheless, by its very definition, the Lindbladian framework is fundamentally limited to memoryless dynamics. Capturing non-Markovian effects ubiquitously observed on spin-boson and superconducting systems instead requires alternative methods such as memory kernel or neural network-based approaches \cite{li_probing_2025, scarpetta_machine_2025}.
 
Regarding broader theoretical applications, Lindbladian learning presents a potential candidate for achieving quantum advantage in Markov random field learning. Markov random fields (MRFs) represent the dependency structure of high-dimensional probability distributions via graphical models, where edges encode conditional dependencies between random variables and joint distributions are defined by clique potentials on the vertices \cite{chow_approximating_1968}. Classically, learning $k$-local MRFs from i.i.d. samples is constrained by the noisy parity barrier, necessitating a runtime of $n^{\Theta(k)}$ \cite{hamilton_information_2017, bresler_structure_2014}. While recent breakthroughs have demonstrated that sampling from Glauber dynamics can bypass this barrier to achieve learning in time $\mathcal O(n^2 \log n)$ assuming constant locality \cite{gaitonde_bypassing_2025}, these classical methods struggle as the interactions become non-local or highly complex. By leveraging the ability to learn sparse, nonlocal Lindbladians efficiently, one could potentially reconstruct the generators of these dynamics efficiently. This offers a pathway to learning nonlocal MRFs that may be classically intractable.

Despite these applications and theoretical promises, there remain several important technical challenges. In particular, ansatz-free Lindbladian learning appears to be bottlenecked by Hamiltonian support learning, which requires $1/\epsilon^4$ scaling of evolution time to learn due to the dominance of the dissipative terms and the inability to directly suppress them. Consequently, it remains an open question whether the Hamiltonian terms can be learned more efficiently in the presence of Lindbladian noise \cite{ivashkov_ansatz-free_2026}. Given that rescaling and inverting dynamics enables Heisenberg-limited scaling of noiseless ansatz-free Hamiltonian learning \cite{hu_ansatz-free_2025}, it remains to be seen whether techniques for approximately inverting Lindbladian dynamics, such as the continuous Petz recovery map, can provide analogous improvements in the noisy setting \cite{kwon_reversing_2022}.

Similar challenges also arise in the SPAM-robust procedure, where the need to simultaneously estimate the coefficients to high precision and have the error be sufficiently small produces an evolution time on the order of $1/\epsilon^8$. While this could be reduced to $1/\epsilon^3$ by assuming knowledge of the Hamiltonian support a priori, it remains to be seen whether we can learn the Hamiltonian terms in less time using techniques such as Prony's method. Moreover, while we outline a universal learning algorithm that works for all Lindbladians, there could be certain coefficients that are still learnable for large classes of Lindbladians. Consequently, it remains an open question how these components can be learned, and how exactly gauge-dependence restricts the applicability of these approaches.

Future research directions will focus on resolving these open questions and addressing the practical constraints of near-term quantum hardware. First, an important objective is to investigate the aforementioned approaches for improving the scaling of Hamiltonian learning in noisy settings. Specifically, it remains to be seen whether approximate inversion techniques, spectral estimation methods such as Prony's method, or alternative dynamical protocols can reduce the evolution-time requirements for learning Hamiltonian components. 

Beyond these error-scaling optimizations, another direction of practical interest is investigating the specific gauge degrees of freedom that arise when the experimenter's control is strictly restricted to noiseless single-qubit unitaries. While we assumed access to arbitrary ideal $n$-qubit control operations for our analysis, multi-qubit gates suffer from poor fidelity in practice and are deeply limited by hardware topology. Characterizing the emergent gauge degrees of freedom under a restricted single-qubit control architecture will therefore provide a more experimentally relevant notion of learnability.

\noindent\textit{Note.}
During the completion of this thesis, we became aware of concurrent independent work by \cite{romanov_learning_2026}, which also explores ansatz-free Lindbladian learning, with a focus on using adaptive quantum error correction to recover Heisenberg-limited scaling for learning the Hamiltonian component and standard-quantum-limited scaling for learning the entire Lindbladian.

\newpage

\bibliographystyle{unsrt}
\bibliography{references}

\appendix

\chapter{Noiseless Learning Error and Complexity Analysis}
\label{chap:support_learn_proofs}

In this chapter, we will rigorously prove the central results of the Lindblad learning algorithm in Chapter \ref{chap:noiseless_lindblad_learn}. We begin by deriving the complexity results for support learning, after which we then prove the results for coefficient learning.

\section{Support Learning Sample Complexity and Time Evolution}

In the main text, we developed a method to learn the support of the Lindbladian by performing Bell sampling at different time resolutions and checking what Paulis we sample. The aim of this section is to rigorously justify the support learning complexities claimed in Proposition \ref{prop:support_learn} and Corollary \ref{cor:dissipator_support_learn}.

\propsupportlearn*

\cordissipatorsupportlearn*

\begin{proof}
    To find the optimal evolution time, we must first perform error analysis on the probabilities of sampling Paulis in the support of the Lindbladian. We begin by analyzing the support of the jump operators. Analyzing the Taylor series error, we have that
    \begin{align*}
        \varepsilon_\text{Taylor} &\le \mathcal O\left(\frac{t^2}{2}\langle \Psi|(I \otimes P_k)(I \otimes \mathcal L^2)[|\Psi\rangle \langle \Psi|](I \otimes P_k)|\Psi\rangle\right) \\
        &= \mathcal O\left(\frac{t^2}{2}\sum_{a, b}\chi_{ab}\chi_{a'b'}\langle \Psi|(I \otimes P_k)(I \otimes P_aP_{a'})|\Psi\rangle \langle \Psi|(I \otimes P_{b'}P_b)(I \otimes P_k)|\Psi\rangle\right) \\
        &= \mathcal O\left(\frac{t^2}{2}\sum_{a, b}\chi_{ab}\chi_{a'b'}\langle \Psi|(I \otimes P_kP_aP_{a'})|\Psi\rangle \langle \Psi|(I \otimes P_{b'}P_bP_k)|\Psi\rangle\right)
    \end{align*}
    Since there are at most $\lfloor M / 2 \rfloor$ Pauli pairs $P_j, P_j' \in \text{supp}(\mathcal L)$ s.t. $P_jP_j' \propto P_k$, we have that assuming $\epsilon \le |\chi_{ab}| \le 1$,
    $$\varepsilon_\text{Taylor} \le \mathcal O(M^2t^2)$$
    for $t \ll 1$. In order for the error to be deemed insignificant, it must affect the probability of sampling a given Pauli by at most a constant factor times the probability. Since the worst-case probability is $\text{Pr}[P_k] \approx \epsilon t$, we require
    $$\varepsilon_\text{Taylor} \le \mathcal O(M^2t^2) \le \frac{\epsilon t}{C} \implies t \le \frac{\epsilon}{CM^2}$$
    for some large constant $C \gg 1$. Since the $\ell$th degree diagonal term in the time-evolution operator can receive contributions from at most $\mathcal O(M^{2(\ell - 1)})$ terms, we have that any choice of $t < \mathcal O(1/M^2)$ is sufficient for the contribution from the higher-order terms to be asymptotically insignificant.
    
    Hence, choosing $t = \Theta(\epsilon/M^2)$, we have that
    $$\text{Pr}[P_k] \ge \Omega(\epsilon t) = \Omega\left(\frac{\epsilon^2}{M^2}\right)$$
    We now find the optimal sample complexity to cover the entire support of the Lindblad jump operators with high probability. Using a coupon-collector style argument, we have that the probability of never sampling one of the Paulis in the support is at most
    $$\text{Pr}[\text{missing Pauli}] \le M\left(1 - \Omega\left(\frac{\epsilon^2}{M^2}\right)\right)^N \le \delta$$
    Hence, we have that to succeed with probability at least $1 - \delta$, we require sample complexity
    $$N = \mathcal O\left(\frac{M^2\log (M/\delta)}{\epsilon^2}\right)$$
    Since we evolve for time $t = \mathcal O(\epsilon/M^2)$, we have that the total time evolution required is
    $$T = \mathcal O \left(\frac{\log (M/\delta)}{\epsilon}\right)$$
    Consequently, we can learn the entire support of the jump operators in Heisenberg-limited evolution time, up to logarithmic factors in $M$.
    
    We can now perform the same analysis for the Hamiltonian terms that are not supported by the jump operators. Just as before, we consider the leading order error term, which, using the same reasoning as before, scales as
    $$\varepsilon_\text{Taylor}^H \le \mathcal O(M^4t^3)$$
    Since $\Pr[P_k] \approx \frac{h_k^2t^2}{2}$, we have that the worst-case error is given as $\Pr[P_k] \approx \frac{\epsilon^2t^2}{2}$. Consequently, we have that for some large constant $C_H \gg 1$, we require
    $$\varepsilon_\text{Taylor}^H \le \mathcal O(M^4t^3) \le \frac{\epsilon^2t^2}{2} \implies t \le \mathcal O\left(\frac{\epsilon^2}{M^4}\right)$$
    Hence, we have that the probability of sampling a Hamiltonian element is lower bounded as
    $$\Pr[P_k] \ge \Omega\left(\frac{\epsilon^6}{M^8}\right)$$
    Consequently, repeating the same coupon-collector proof from before, we have that the required sample complexity is
    $$N = \mathcal O\left(\frac{M^8 \log (M / \delta)}{\epsilon^6}\right)$$
    Given the size of each time-step, we therefore have that the total evolution time required is given as
    $$T = \mathcal O\left(\frac{M^4 \log(M/\delta)}{\epsilon^4}\right)$$
    As mentioned previously, we will assume $\delta \sim \mathcal O(1)$ to be a very small constant. 
\end{proof}

\section{Oversampling the Support}

Having established these complexities, we now demonstrate that oversampling only introduces logarithmic overhead via Lemma \ref{lem:support_size}.

\lemsupportsize*

\begin{proof}
    We begin by analyzing the number of non-identity Paulis we sample for the dissipator support. From Algorithm \ref{alg:dissipator_support_learn} we have that the sample complexity is chosen to be
    $$N = \mathcal O\left(\frac{M^2 \log (M/\delta)}{\epsilon^2}\right)$$
    Based on the Taylor expansion we performed, we have that the sum of the Taylor error bounds for each individual Pauli in the support of the dissipator forms an upper bound on the probability of sampling outside the support of the dissipator. In particular, since $\varepsilon_\text{Taylor} \le \mathcal O(M^2t^2)$ for each Pauli, we have that the probability of sampling outside the support is upper bounded as
    $$p_\text{out} \le \mathcal O(M^3t^2) = \mathcal O(\epsilon^2/M)$$
    Consider a set of indicators $\{X_i\}_{i \in [N]}$ where $X_i = 1$ if sample $i$ lies outside the support and 0 otherwise. We have from the multiplicative Chernoff bound that
    $$\Pr\left(\sum_{i = 1} ^N X_i \ge \alpha Np_\text{out}\right) \le \exp\left(-\frac{(\alpha - 1)^2Np_\text{out}}{\alpha + 1}\right)$$
    Since $Np_\text{out} \le \Theta(M \log M)$, we have that
    $$\Pr\left(M' \ge \alpha \cdot \Theta(M \log M)\right) \le \exp\left(-\Theta(\alpha M \log M)\right),$$
    where $M' = \sum_i X_i$ is the number of Paulis sampled outside the support. Consequently, we have that with high probability, $M' \le \mathcal O(\alpha M \log M)$ for some sufficiently large constant $\alpha$. Consequently, we have that
    $$|\hat {\mathcal S}| \le M + M' = \mathcal O(M \log M)$$
    with high probability.
    
    For the Hamiltonian support, we have that
    $$N_H = \mathcal O\left(\frac{M^8 \log (M/\delta)}{\epsilon^6}\right)$$
    Here, we have that since the second-order terms in the Lindbladian are all in $\mathcal S_H$, we simply care about the third-order error term, meaning
    $$p_\text{out, $H$} \le \mathcal O(M^5t_H^3) = \mathcal O(\epsilon^6/M^7)$$
    Performing the same Chernoff analysis as before with $N_Hp_\text{out, $H$} \le \mathcal O(M \log M)$, we get the same bound for the number $M'_H$ of Paulis sampled outside the true support set $\mathcal S_H$:
    $$\Pr\left(M'_H \ge \alpha \cdot \Theta(M \log M)\right) \le \exp\left(-\Theta(\alpha M \log M)\right),$$
    Consequently, we have that
    $$|\hat{\mathcal S}_H| \le |\hat {\mathcal S}|^2 + M_H' \le |\hat {\mathcal S}|^2 + \mathcal O(M \log M)$$
    with high probability.
\end{proof}

\section{Diagonal Coefficient Learning Complexity}

Having learned the support, we now focus on the evolution time required to learn the coefficients. Starting with the diagonal coefficients, we now rigorously recover the time complexity stated in Proposition \ref{prop:on_diagonal}.

\propondiagonal*

\begin{proof}
    We previously observed that by evolving for time $t = \mathcal O(\epsilon/M^2)$, the probability of sampling a given diagonal coefficients $p_{kk}$ differs from its linear order approximation $\chi_{kk}t$ by at most $\frac{\epsilon^2}{CM^2}$ for some large constant $C \gg 1$. Since we have to rescale $\chi_{kk}t$ by $\frac{M^2}{\epsilon}$, we have that with high probability,
    $$|\tilde \chi_{kk} - \chi_{kk}| \le \frac{\epsilon}{C},$$
    where $\tilde \chi_{kk}$ is the true empirical coefficient. Consequently, if we can obtain an estimate $\hat \chi_{kk}$ for $\tilde \chi_{kk}$ with accuracy threshold at most $\left(1 - \frac{1}{C}\right)\epsilon$, we can get
    $$|\hat \chi_{kk} - \chi_{kk}| \le \epsilon,$$
    with high probability, as desired. Hence, this procedure allows us to learn these coefficients to $\epsilon$-accuracy. 
    
    Since we know that this time-resolution is sufficient to learn this coefficient to $\epsilon$-accuracy, we now consider the sample complexity necessary to achieve this. Since learning each coefficient to $\epsilon$-accuracy implies learning the individual probabilities $\chi_{kk}t$ to $\mathcal O(\epsilon^2/M^2)$ accuracy, Hoeffding's inequality states that
    $$\Pr\left(|\hat \chi_{kk} - \chi_{kk}|t \ge \frac{\epsilon^2}{C'M^2}\right) \le 2\exp\left(-2N\frac{\epsilon^4}{C'^2M^4}\right)$$
    Taking a union bound over all $k \in \hat {\mathcal S}$, we have that
    $$\Pr\left(\forall k: |\hat \chi_{kk} - \chi_{kk}|t \ge \frac{\epsilon^2}{C'M^2}\right) \le \mathcal O(M \log M) \cdot \exp\left(-2N\frac{\epsilon^4}{C'^2M^4}\right) \le \delta$$
    Consequently, we have to learn all the diagonal coefficients to $\epsilon$-accuracy with probability at least $1 - \delta$, we require sample complexity
    $$N = \mathcal O\left(\frac{M^4 \log (M/\delta)}{\epsilon^4}\right)$$
    With our choice of time resolution, this means that the total evolution time required for this step is
    $$T = \mathcal O\left(\frac{M^2 \log (M/\delta)}{\epsilon^3}\right)$$
\end{proof}

\section{Off-Diagonal Coefficient Complexity}

Now that we have learned the diagonal coefficients, we focus on the off-diagonals. For the off-diagonal analysis, we begin proving that the normalization factor is close to $1/2$, as stated in Lemma \ref{lem:normalization}.

\lemnormalization*

\begin{proof}
    We had that the renormalization factor as
    $$\kappa_{ij}^\pm = \text{Tr}(K^\pm_{ij} e^{\mathcal Lt}[\rho] (K^\pm_{ij})^\dag)),$$ 
    which corresponds to the probability of applying $K_{ij}^+$ and $K_{ij}^-$, respectively. We can write these probabilities in term of expectation values as follows. Using (\ref{eq:kraus_normalization}), we have that $\kappa_{ij}^\pm$ must be one of either $\frac{1 + \langle P_k\rangle}{2}$ or $\frac{1 - \langle P_k\rangle}{2}$. 
    
    In order to bound this expectation value, we use the previously derived fact that for Pauli $P_k$,
    $$\text{Tr}((\mathcal E - \mathcal E')[\rho] P_k) \le \|\mathcal E - \mathcal E'\|_\diamond$$
    Taking $\mathcal E = e^{t\mathcal L}$ and $\mathcal E' = \mathcal I$, this means that we can bound the distance between the expectation value before and after evolution as $\|\mathcal e^{\mathcal Lt} - \mathcal I\|_\diamond$. Since $\mathcal L$ is $M$-sparse, we have that $\chi(\mathcal L)$ as at most $\mathcal O(M^2)$ nonzero elements, all of which have magnitude at most 1, giving us that
    $$\|\mathcal L\|_\diamond \le \mathcal O(M^2) \implies \|e^{t\mathcal L} - \mathcal I\|_\diamond \le e^{\mathcal O(tM^2)} - 1$$
    Taking $t = \mathcal O(\epsilon/M^2)$ means that this bound effectively becomes $O(\epsilon)$. For any Bell state, we have that $\langle P_k \rangle = 0$ for $P_k \ne I$, as operating with a Pauli on one side of the Bell pairs produces an orthogonal $n$-qubit Bell state. Consequently, this means that
    $$|\text{Tr}(e^{t\mathcal L}[\rho]P_k)| = \mathcal O(\epsilon)$$
    Consequently, we have that
    $$\left|\kappa_{ij}^\pm - \frac{1}{2}\right| \le \mathcal O(\epsilon)$$
\end{proof}

Having shown that $\kappa_{ij}^\pm$ is close to $1/2$, we now show that the evolution time required matches that of Proposition \ref{prop:off_diagonal}.

\propoffdiagonal*

\begin{proof}
    To ensure that we estimate $\kappa_{ij}^\pm$ to the desired accuracy, we perform the following analysis. We have that the probability $p_k$ of sampling a given Pauli $P_k$ can be estimated to the following accuracy:
    $$|\hat \kappa_{ij}^\pm \hat p_k - \kappa_{ij}^\pm p_k| \le |\hat \kappa_{ij}^\pm - \kappa_{ij}^\pm|\hat p_k + \kappa_{ij}^\pm|\hat p_k - p_k| \le \mathcal O(\epsilon t)$$
    Assuming we measure $\hat p_k$ to $\mathcal O(\epsilon t)$-accuracy, this means that taking $\hat p \sim \mathcal O(1)$ in the maximal case implies that a sufficient accuracy to estimate $\hat \kappa_{ij}^\pm$ is $\mathcal O(\epsilon t) = \mathcal O(\epsilon^2/M^2)$. Consequently, the number of samples originally used to estimate the coefficients is itself sufficient to estimate the probabilities of choosing each ancilla qubit to the desired accuracy.
    
    Altogether, we have that each off-diagonal coefficient can be learned using
    $$N = \mathcal O\left(\frac{M^4 \log (1/\delta)}{\epsilon^4}\right)$$
    samples, or
    $$T = \mathcal O\left(\frac{M^2 \log (1/\delta)}{\epsilon^3}\right)$$
    evolution time. Since we have to learn $\mathcal O(M^2)$ off-diagonals, we have that the total sample complexity and evolution time are given as follows:
    $$N = \mathcal O\left(\frac{M^6 \log(1/\delta)}{\epsilon^4}\right)$$
    $$T = \mathcal O\left(\frac{M^4 \log(1/\delta)}{\epsilon^3}\right)$$
\end{proof}

\chapter{Gauge Degree of Freedom Analysis}
\label{chap:gauge}

In this chapter, we will rigorously prove the unlearnability results in Chapter \ref{chap:spam_lindblad_learn}. We begin by deriving the transformed error channels before working with the transformed Lindbladians.

\section{Error Channel Gauge Transformation}

We now set about proving Lemma \ref{lem:gauge_error_channel} by showing that there exist transformed error channels that satisfy the desired conditions.

\lemgaugeerrorchannel*

\begin{proof}
    As stated before, we consider the transformation 
    \begin{align*}
        \mathcal E_1 &\mapsto \mathcal E_1' = e^{-\theta \Gamma} \circ \mathcal E_1 \\
        \mathcal E_2 &\mapsto \mathcal E_2' = \mathcal E_2 \circ e^{\theta \Gamma}
    \end{align*}
    In the $n$-qubit case, we define the generator $\Gamma$ by
    $$\Gamma[\rho] = \frac{1}{4^n}\left[(4^n - 1)\rho - \sum_{i \ne 0}P_i \rho P_i\right]$$
    Computing the Pauli transfer matrix of this generator, we have that
    $$\mathcal R(\Gamma) = \begin{pmatrix}
        0 & 0 \\
        0 & \mathbf I
    \end{pmatrix} \implies \mathcal R(e^{-\theta \Gamma}) = \begin{pmatrix}
        1 & 0 \\
        0 & e^{-\theta}\mathbf I
    \end{pmatrix}$$
    For this mapping to produce a gauge transformation, we must have that the transform channels are still CPTP and that they still satisfy the SPAM-robustness criterion:
    $$\| \mathcal E_1' - \mathcal I \|_\diamond + \| \mathcal E_2' - \mathcal I \|_\diamond \le \varepsilon_\text{SPAM}$$
    Since our error channels are uniformly depolarizing, we can express them as follows:
    $$\mathcal E_i[\cdot] = (1 - p_i)[\cdot] + p_i\text{Tr}[\cdot]\frac{I}{2^n}$$
    We can bound the diamond distance between $\mathcal E_i[\rho]$ and $\mathcal I$ as follows:
    $$\|\mathcal E_i - \mathcal I\|_\diamond = 2p_i\left(1 - \frac{1}{4^n}\right) \le 2p_i$$
    Consequently, we can choose $p_1 + p_2 \le \frac{\varepsilon_\text{SPAM}}{2}$ to ensure that the SPAM constraint is satisfied. We now check when $\mathcal E_1', \mathcal E_2'$ are CPTP. We have that
    $$\mathcal R(\mathcal E_i) = \begin{pmatrix}
        1 & 0 \\
        0 & (1 - p_i)\mathbf I
    \end{pmatrix}$$
    Hence, we have that
    \begin{align*}
        \mathcal R(\mathcal E_1') &= \begin{pmatrix}
            1 & 0 \\
            0 & e^{-\theta}(1 - p_1)\mathbf I
        \end{pmatrix} \\
        \mathcal R(\mathcal E_2') &= \begin{pmatrix}
            1 & 0 \\
            0 & e^{\theta}(1 - p_2)\mathbf I
        \end{pmatrix}
    \end{align*}
    Hence, we have that this transformation simply yields new depolarizing channels with probabilities $p_1', p_2'$ given as follows:
    \begin{align*}
        p_1' &= 1 - e^{-\theta}(1 - p_1) \\
        p_2' &= 1 - e^{\theta}(1 - p_2)
    \end{align*}
    Here, we impose two constraints on $p_1'$ and $p_2'$. Firstly, we need to satisfy the SPAM constraint, meaning that
    $$p_1' + p_2' = 2 - e^{-\theta}(1 - p_1) - e^{\theta}(1 - p_2) \le \frac{\varepsilon_\text{SPAM}}{2}$$
    Secondly, while $0 \le p_1' \le 1$, we have that $p_2'$ can be negative if $\theta$ is left unconstrained. Hence, we must have that
    $$p_2' = 1 - e^{\theta}(1 - p_2) \ge 0 \implies \theta \le -\ln(1 - p_2)$$
    Using the first derivation condition, we find that the $\theta$ that maximizes the LHS of the first inequality is given as:
    $$\theta^* = \frac{1}{2}\ln\left(\frac{1 - p_1}{1 - p_2}\right)$$
    Substituting this into the original expression, we have that
    \begin{align*}
        p_1' + p_2' &\le 2 - e^{-\theta^*}(1 - p_1) - e^{\theta^*}(1 - p_2) \\ 
        &= 2 - e^{-\frac{1}{2}\ln\left(\frac{1 - p_1}{1 - p_2}\right)}(1 - p_1) - e^{\frac{1}{2}\ln\left(\frac{1 - p_1}{1 - p_2}\right)}(1 - p_2) \\
        &= 2 - \sqrt{\frac{1 - p_2}{1 - p_1}}(1 - p_1) - \sqrt{\frac{1 - p_1}{1 - p_2}}(1 - p_2) \\
        &= 2 - 2\sqrt{(1 - p_1)(1 - p_2)} \\
        &\le \frac{\varepsilon_\text{SPAM}}{2}
    \end{align*}
    Hence, we must have that
    \begin{align*}
        2 - 2\sqrt{(1 - p_1)(1 - p_2)} &\le \frac{\varepsilon_\text{SPAM}}{2} \\
        (1 - p_1)(1 - p_2) &\ge \left(1 - \frac{\varepsilon_\text{SPAM}}{4}\right)^2 \\
        1 - p_1 - p_2 + p_1p_2 &\ge 1 -  \frac{\varepsilon_\text{SPAM}}{2} + \frac{\varepsilon_\text{SPAM}^2}{16} \\
        p_1 + p_2 - p_1p_2 &\le \frac{\varepsilon_\text{SPAM}}{2} - \frac{\varepsilon_\text{SPAM}^2}{16}
    \end{align*}
    We can observe that this bound is saturated by choosing $p_1 = p_2 = \varepsilon_\text{SPAM}/4$, which also satisfies the necessary constraint for the original channel to be SPAM-robust. For this choice of $p_1, p_2$, we have that
    $$0 \le \theta \le -\ln(1 - \varepsilon_\text{SPAM}/4)$$
    Hence, for this choices of parameters, this transformation produces valid quantum channels $\mathcal E_1', \mathcal E_2'$ that satisfy the SPAM-robust requirement.
\end{proof}

\section{Single-Qubit Lindbladian Gauge Transformation}

Having shown the validity of the error channel, we now demonstrate the validity of the transformed Lindbladian. Starting with the single-qubit case, we now prove Lemma \ref{lem:gauge_single_lindbladian}.

\lemgaugesinglelindbladian*

\begin{proof}
    Recall that for any operator $A$,
    $$\text{Tr}[\mathcal L(A)] = 0$$
    Hence, we have that the Pauli transfer matrix of the Lindbladian is given as
    $$\mathcal R(\mathcal L) = \begin{pmatrix}
        0 & 0 \\
        \mathbf r & \mathbf M
    \end{pmatrix}$$
    Consequently, we have that
    $$\mathcal R(e^{-\theta \Gamma}\mathcal Le^{\theta \Gamma}) = \begin{pmatrix}
        0 & 0 \\
        e^{-\theta}\mathbf r & \mathbf M
    \end{pmatrix}$$
    In other words, we have that
    $$\mathcal L'[\rho] = \mathcal L[\rho] - (1 - e^{-\theta})(\mathbf r \cdot \mathbf P)\mathcal D[\rho],$$
    where $\mathbf P = (X, Y, Z)$ and $\mathcal D$ is a uniform depolarizing channel. We can show that this PTM corresponds to a Lindbladian geometrically. Computing the PTM for $e^{t\mathcal L}$, we have that
    $$\mathcal R(e^{t\mathcal L}) = \begin{pmatrix}
        1 & \mathbf 0 \\
        \int_0 ^t e^{(t - s)\mathbf M} \mathbf r \; ds & e^{t\mathbf M}
    \end{pmatrix}$$
    Let us denote $\mathbf d(t) = \int_0 ^t e^{(t - s)\mathbf M} \mathbf r \; ds$. In general, we have that a CPTP map transforms the Bloch sphere into a rotated ellipse that lies entirely within said sphere, with the center being located at $\mathbf d(t)$ \cite{king_minimal_2001}. Consequently, shifting the center to $e^{-\theta}\mathbf d(t)$ means that the ellipsoid still lies within the Bloch sphere. Hence, this is a valid Lindbladian.
\end{proof}

\section{Multi-Qubit Lindbladian Gauge Transformation}

Since there exist Lindbladians that are not mapped to valid Lindbladians under this type of gauge transformation, we instead focus on proving the existence of Lindbladians that can be transformed. To this end, we prove Lemma \ref{lem:gauge_multi_lindbladian} as follows:

\lemgaugemultilindbladian*

\begin{proof}
    Suppose we have a Lindbladian $\mathcal L$ where
    $$\mathcal L(I) = R$$
    for some $R \in \mathcal P_n$. We have that in the PTM representation, this element can be expressed as the following operator expansion by performing left multiplication on a uniformly depolarizing channel:
    $$|R\rangle\rangle\langle\langle I| = \frac{1}{4^n}\sum_{j}RP_j\rho P_j$$
    Consequently, if we want $\mathcal L'(I) = e^{-\theta}R$, we can subtract this expression from the original Lindbladian giving us the following:
    $$\mathcal L'[\rho] = \mathcal L[\rho] - \frac{1 - e^{-\theta}}{4^n}\sum_{j}RP_j\rho P_j$$
    To ensure that $\mathcal L'$ is a Lindbladian, we can have the coefficients for all the terms being subtracted be equal to $\frac{1}{4^n}$ in the original $\mathcal L$, as this means that choosing $\theta > 0$ means we effectively supress these terms exponentially, with them going to 0 as $\theta \to \infty$.

    To get terms of the form $RP_j\rho P_j$, we can choose jump operators of the form
    $$L_j = \frac{1}{2^{n + 1/2}}(P_j + RP_j)$$
    Note that these jump operators need not be unique. In particular, consider $P_k = \omega RP_j$. We have that
    $$L_k = \frac{1}{2^{n + 1/2}}(P_k + RP_k) = \frac{\omega}{2^{n + 1/2}}(P_j + RP_j) = \omega L_j$$
    Hence, we have that $L_k$ and $L_j$ are effectively identical jump operators, as they differ only by a phase. Consequently, to account for this double counting, we include the additional factor of $1/\sqrt{2}$. Consequently, we have that
    $$L_j\rho L_j^\dag + L_k \rho L_k^\dag = \frac{1}{4^n}(P_j + RP_j)\rho (P_j + P_jR) = \frac{1}{4^n}[P_j \rho P_j + RP_j \rho P_j + RP_k \rho P_k + P_k \rho P_k]$$
    Hence, we have that these $L_j, L_k$ jump operators collectively capture both the $RP_j \rho P_j$ and $RP_k \rho P_k$ terms. Hence, if we now iterate over all the Pauli terms with this construction, we would have that
    $$\sum_j L_j \rho L_j^\dag = \frac{1}{4^n}\sum_j[P_j\rho P_j + RP_j \rho P_j]$$
    Hence, the gauge transformation simply suppresses the second term here. Nevertheless, we have not yet accounted for the normalization terms $-\frac{1}{2}\{L_j^\dag L_j, \rho\}$. Evaluating $L_j ^\dag L_j$, we have that
    $$L_j^\dag L_j = \frac{1}{2^{2n + 1}}(P_j + P_jR)(P_j + RP_j) = \frac{1}{4^n}(I + P_jRP_j)$$
    We note that $P_jRP_j = \pm R$ depending on whether $P_j$ and $R$ commute or not. Since $R$ commutes with exactly half of the Paulis, we have that
    $$-\frac{1}{2}\sum_j \{L_j^\dag L_j, \rho\} = -\frac{1}{2}\{I, \rho\} = -\rho$$
    Hence, we have that the normalization term only affects the on-diagonal of the Chi matrix. Altogether, this means that which this choice of jump operators, the entire Lindbladian $\mathcal L$ is given as
    $$\mathcal L[\rho] = \frac{1}{4^n}\sum_j[P_j\rho P_j + RP_j \rho P_j] - \rho$$
    Consequently, we have that the transformed Lindbladian is given as
    $$\mathcal L'(\theta)[\rho] = \frac{1}{4^n}\sum_j[P_j\rho P_j + e^{-\theta}RP_j \rho P_j] - \rho$$
    Showing that this particular $\mathcal L'$ is still a Lindbladian follows in straightforward fashion. In particular, taking $\theta \to \infty$, we have that
    $$\mathcal L'(\infty)[\rho] = \frac{1}{4^n}\sum_j[P_j\rho P_j - \rho],$$
    which is simply a Pauli Lindbladian. We can then write $\mathcal L'(\theta)[\rho]$ as the following convex combination:
    $$\mathcal L'(\theta)[\rho] = e^{-\theta}\mathcal L[\rho] + (1 - e^{-\theta})\mathcal L'(\infty)[\rho]$$
    Since $\mathcal L(\theta)[\rho]$ is a convex combination of Lindbladians for all $\theta \ge 0$, we have that $\mathcal L'(\theta)[\rho]$ is a valid Lindbladian.

    We now find the components of this Lindbladian that are gauge-dependent. In particular, we can see that when $[R, P_j] = 0$, we have that $RP_j = \pm P_k$ for some $P_k \in \mathcal P_n$. Hence, we have that the term $RP_j \rho P_k + RP_k \rho P_k$ affects the real component of $\chi_{jk}$. On the other hand, if $\{R, P_j\} = 0$, we have that $RP_j = \pm iP_k$. Consequently, we have that $RP_j \rho P_k + RP_k \rho P_k$ affects the imaginary component of $\chi_{jk}$. Since the commutation relation between $P_j$ and $R$ is identical to that between $P_j$ and $P_k$, we have that $\text{Re}[\chi_{jk}]$ for commuting $P_j, P_k$ and $\text{Im}[\chi_{jk}]$ for anticommuting $P_j, P_k$ are the gauge-dependent coefficients.

    Furthermore, since we can perform this Lindbladian construction for all $R \in \mathcal P_n$, we have that this extends to all choices of $P_j, P_k$.
\end{proof}

\chapter{SPAM-Robust Learning Error Analysis}\label{chap:spam_lindblad_proofs}

Having demonstrated which components are unlearnable, we now consider how to learn the remaining components in evolution time $T = \text{poly}(M, n, 1/\epsilon)$. In this section, we will prove all the results pertaining to SPAM-robust learning presented in Chapter \ref{chap:spam_lindblad_learn}.

\section{Twirled Lindblad Evolution}

To bound the error in simulating evolution under the twirled Lindbladian, we prove Lemma \ref{lem:lindblad_pauli_twirl}.

\lemlindbladpaulitwirl*

\begin{proof}
    Using a telescoping sum argument, we have that
    $$\|\mathcal E'_{t, r} - \mathcal E_{t}\|_\diamond \le r\|\mathcal E'_{t/r, 1} - \mathcal E_{t/r}\|_\diamond$$
    Expanding this out to second order, we have that
    $$\|\mathcal E'_{t/r, 1} - \mathcal E_{t/r}\|_\diamond = \|\mathbb E_{P \sim \mathcal P_n}[Pe^{t\mathcal L/r}[P\cdot P]P] - e^{t\mathcal L_\text{tw}/r}\|_\diamond = \mathcal O\left(\frac{t^2}{r^2}\|(\mathcal L^2)_\text{tw} - \mathcal L_\text{tw}^2\|_\diamond\right)$$
    In order to bound this diamond norm term, we use the fact that for a superoperator $\mathcal A: X \mapsto AX$, $\|\mathcal A\|_\diamond \le \|A\|_\infty$. Consequently, we have that for any superoperator $\mathcal B[\rho] = \sum_{i, j}\chi_{ij}P_i \rho P_j$, we have  that
    $$\|\mathcal B\|_\diamond \le \sum_{i, j}|\chi_{ij}|$$
    Since each element of the chi matrix of $\mathcal L$ has absolute value at most 1 and at most $\mathcal O(M^2)$ terms contribute to each element of the chi matrix of $\mathcal L^2$, we have that $\|(\mathcal L^2)_\text{tw} \|_\diamond \le \mathcal O(M^3)$ and $\|(\mathcal L_\text{tw})^2 \|_\diamond \le \mathcal O(M^2)$.
    Hence, we have altogether that
    $$\|\mathcal E'_{t, r} - \mathcal E_{t}\|_\diamond \le \mathcal O\left(\frac{M^3t^2}{r}\right)$$
    Note that since $k$th order error terms scale as $\frac{t^kM^{2k + 1}}{r^k}$, we have that this bound holds for $t/r \le \mathcal O(M^{-2})$
\end{proof}

\section{On-Diagonal Uncertainty Analysis}

Having bounded the Trotterization error, we now determine the accuracy to which we need to measure the exponential eigenvalue in order to recover the exponent to the desired accuracy. In the main text, we accomplished this using Lemma \ref{lem:eigen_uncertainty}. We now prove this as follows:

\lemeigenuncertainty*

\begin{proof}
    For $f = e^{\gamma t}$, we have that to estimate $\gamma$ to accuracy on the order of $\eta$, we need to estimate $f$ to accuracy
    $$\delta f \approx \left|\frac{\delta f}{\delta \gamma}\right|\delta \gamma \sim \mathcal O\left(\eta te^{\gamma t}\right)$$
    This means that if we evolve for times $t$ and $2t$, we need to estimate $\frac{f_2}{f_1}$ to accuracy $\mathcal O(\eta te^{\gamma t})$. We have that
    $$\left|\frac{\delta (f_2/f_1)}{f_2/f_1}\right| \approx \frac{\delta f_2}{f_2} + \frac{\delta f_1}{f_1} \sim \eta t \implies \delta f_1 \sim \mathcal O(A\eta te^{\gamma t}), \delta f_2 \sim \mathcal O(A\eta te^{2\gamma t})$$
\end{proof}

Since $A = \lambda_{1, Q}\lambda_{2, Q}$, we now demonstrate how to lower bound these quantities via Lemma \ref{lem:spam_eigenvalue_bound}. To accomplish this, we first prove the following helper lemma:
\begin{lemma}[Twirling Diamond Distance]
    Given a CPTP map $\mathcal E$, and its Pauli-twirled version $\mathcal E_\text{tw}[\cdot] = \frac{1}{4^n}\sum_i P_i \mathcal E[P_i \cdot P_i]P_i$, we have that
    $$\| \mathcal E_\text{tw} - \mathcal I \|_\diamond \le \| \mathcal E - \mathcal I \|_\diamond,$$
    where $\mathcal I$ is the identity channel.
\end{lemma}
\begin{proof}
    Using triangle inequality, the invariance of diamond distance under unitary transformations, and the fact that $\mathcal I$ is invariant under twirling, we have that
    \begin{align*}
        \| \mathcal E_\text{tw} - \mathcal I \|_\diamond &= \left\| \frac{1}{4^n}\sum_i P_i \mathcal E[P_i \cdot P_i]P_i - \mathcal I \right\|_\diamond \\
        &= \left\| \frac{1}{4^n}\sum_i P_i (\mathcal E - \mathcal I)[P_i \cdot P_i]P_i \right\|_\diamond \\
        &\le \frac{1}{4^n}\sum_i \left\|P_i (\mathcal E - \mathcal I)[P_i \cdot P_i]P_i \right\|_\diamond \\
        &= \frac{1}{4^n}\sum_i \left\|\mathcal E - \mathcal I \right\|_\diamond \\
        &= \left\|\mathcal E - \mathcal I \right\|_\diamond
    \end{align*}
\end{proof}
With this, we are now ready to prove Lemma \ref{lem:spam_eigenvalue_bound}:

\lemspameigenvaluebound*

\begin{proof}
    For any Pauli channel $\mathcal E = \sum_P p_PP\rho P$, we have that
    $$\|\mathcal E - \mathcal I\|_\diamond = 2(1 - p_I) \implies p_I = 1 - \frac{1}{2}\|\mathcal E - \mathcal I\|_\diamond$$
    since diamond distance is twice the total variation distance of the weights. Since each eigenvalue is a linear combination of the Pauli weights with $+/-$ factors based on commutation relations, we have that the smallest eigenvalue can be lower bounded as
    $$\lambda_\text{min} \ge p_I - \sum_{P \ne I}p_P = 2p_I - 1 = 1 - \|\mathcal E - \mathcal I\|_\diamond$$
    Consequently, for the SPAM error channels $\mathcal E_1', \mathcal E_2'$, we have that
    $$\lambda_{1, Q}\lambda_{2, Q} \ge 1 - \|\mathcal E_1' - \mathcal I\|_\diamond - \|\mathcal E_1' - \mathcal I\|_\diamond \ge 1 - \|\mathcal E_1 - \mathcal I\|_\diamond - \|\mathcal E_1 - \mathcal I\|_\diamond \ge 1 - \varepsilon_\text{SPAM}$$
    Consequently, so long as $\varepsilon_\text{SPAM} < 1$, we have that
    $$|\lambda_{1, Q}\lambda_{2, Q}| \ge 1 - \varepsilon_\text{SPAM}$$
\end{proof}

Choosing $t = \mathcal O(1/M)$ and $\eta = \mathcal O(\epsilon)$, we now can prove Proposition \ref{prop:spam_on_diagonal}:

\propspamondiagonal*

\begin{proof}
    Assuming that $\varepsilon_\text{SPAM} < 1 - \mathcal O(1)$, we have that by Lemmas \ref{lem:eigen_uncertainty} and \ref{lem:spam_eigenvalue_bound}, we just need to learn each eigenvalue to accuracy $\mathcal O(\epsilon/M)$. Revising the Trotter step analysis in Lemma \ref{lem:lindblad_pauli_twirl}, this means that
    $$\frac{M^3t^2}{r} \sim \frac{\epsilon}{M} \implies r \sim \frac{M^2}{\epsilon}$$
    By Hoeffding's inequality, we need to sample
    $$N = \mathcal O\left(\frac{M^2 \log (1 / \delta)}{\epsilon^2}\right)$$
    to estimate each eigenvalue. Since each experiment uses $t = O(1/M)$ time, this means that the time per probe is given as $\mathcal O\left(\frac{M \log (1 / \delta)}{\epsilon^2}\right)$. Intuitively, the number of Pauli probes scales approximately linearly with the size of the support up to logarithmic factors. In particular, in Appendix \ref{chap:linear}, we show that using $K = \mathcal O(Mn\log^2 M)$ probes is sufficient to solve for the coefficients.
    
    Consequently, we have that the overall evolution time required to learn the diagonal coefficients is given as
    $$T = \mathcal O \left(\frac{M^2n\log^2 M}{\epsilon^2}\right)$$
\end{proof}

\section{Off-Diagonal Uncertainty Analysis}

Performing this same analysis for Proposition \ref{prop:spam_off_diagonal}, we have the following:

\propspamoffdiagonal*

\begin{proof}
    We begin by expanding out
    $$(\mathcal C_{ij}\circ e^{\mathcal Lt})[\rho] \approx \mathcal C_{ij}[\rho + t\mathcal L[\rho] + \dots] \approx \frac{1}{2}P_i \rho P_i + \frac{1}{2}P_j \rho P_j + t\sum_{k}\chi_{k}'P_k \rho P_k + (\text{non-diag})$$
    
    For small time $t$, we therefore have that the eigenvalues of the twirl of this are given as
    $$\lambda_Q = \frac{1}{2}[(-1)^{\langle P_i, Q\rangle} + (-1)^{\langle P_j, Q\rangle}] + t\sum_{k}(-1)^{\langle P_k, Q \rangle}\chi_{kk}' = a + bt,$$
    Since the eigenvalue can potentially be extremely small here, we need to instead measure $A = \lambda_{1, Q}\lambda_{2, Q}$ directly without evolving. Denoting estimates $y_1 = A\lambda_Q$ and $y_0 = A$, we have that
    $$\delta \lambda_Q \sim \frac{\delta y_1}{A} + \frac{\delta y_0}{A}|\lambda_Q| = \mathcal O(\epsilon t)$$
    Since $|\lambda_Q| \le 1$, we have that choosing $\delta y_1, \delta y_0 \sim \mathcal O(A \epsilon t)$ is sufficient. Since we require $\epsilon t \sim M^2t^2$, we set $t = \epsilon/M^2$ just as we did for the ideal learning protocol. Hence, we will estimate these eigenvalues to accuracy $\epsilon^2/M^2$. This gives us a sample complexity of
    $$N = \mathcal O\left(\frac{M^4 \log(1/\delta)}{\epsilon^4}\right)$$
    and evolution time of
    $$T = \mathcal O\left(\frac{M^2 \log(1/\delta)}{\epsilon^3}\right)$$
    per Pauli probe. Since the number of coefficients we have to estimate is of the same order as before, we require $K = \mathcal O(Mn \log^2 M)$, giving us a total evolution time per coefficient of
    $$T = \mathcal O\left(\frac{M^3n \log^2 M}{\epsilon^3}\right)$$
    Since we have $\mathcal O(M^2)$ off-diagonals to estimate, this gives us a total evolution time of
    $$T = \mathcal O\left(\frac{M^5n \log^2 M}{\epsilon^3}\right)$$
\end{proof}

\section{Hamiltonian Uncertainty Analysis}

For the Hamiltonian terms, we need to extend the twirling result as demonstrated in Lemma \ref{lem:lindblad_pauli_twirl_quad}.

\lemlindbladpaulitwirlquad*

\begin{proof}
    Using a telescoping sum argument, we have that
    $$\|\mathcal E'_{t, r} - \mathcal E_{t, r}\|_\diamond \le r\|\mathcal E'_{t/r, 1} - \mathcal E_{t/r, 1}\|_\diamond$$
    Expanding this out to third order, we have that
    \begin{align*}
        \|\mathcal E'_{t/r, 1} - \mathcal E_{t/r}\|_\diamond &= \|\mathbb E_{P \sim \mathcal P_n}[Pe^{t\mathcal L/r}[P\cdot P]P] - e^{t\mathcal L_\text{tw}/r + \frac{t^2}{2r^2}[(\mathcal L^2)_\text{tw} - (\mathcal L_\text{tw})^2]}\|_\diamond \\ 
        &\le \mathcal O\left(\frac{t^3M^5}{r^3}\right)
    \end{align*}
    The bound here is obtained from the fact that each of the diagonal elements of $\mathcal L^k$ is of order at most $\mathcal O(M^{2k})$, meaning that twirling the Lindbladian and taking the sum gives a bound of $\mathcal O(M^{2k + 1})$. 
    Hence, we have altogether that
    $$\|\mathcal E'_{t, r} - \mathcal E_{t, r}\|_\diamond \le \mathcal O\left(\frac{M^5t^3}{r^2}\right)$$
\end{proof}

Using the same error analysis we did for Proposition \ref{prop:spam_on_diagonal}, we can compute the evolution time for learning the Hamiltonian terms, as described in Proposition \ref{prop:spam_hamiltonian}.

\propspamhamiltonian*

\begin{proof}
    Substituting $\epsilon \to \epsilon^2t/r$ in the on-diagonal derivation means that we need to estimate the exponential eigenvalues to accuracy $\mathcal O(\epsilon^2t^2/r)$. Since the diamond distance that we computed in Lemma \ref{lem:lindblad_pauli_twirl_quad} bounds the accuracy we can estimate the eigenvalue, we need
    $$\frac{\epsilon^2t^2}{r} \sim \frac{M^5t^3}{r^2} \implies r \sim \frac{M^5t}{\epsilon^2}$$
    Consequently, this means that we need to estimate each eigenvalue to accuracy $\mathcal O(\epsilon^4t/M^5)$. Choosing the maximal $t \sim \mathcal O(1/M)$, this means that the required estimation accuracy is $\mathcal O(\epsilon^4/M^6)$.
    Hence, Hoeffding's inequality gives us a sample complexity of
    $$N = \mathcal O\left(\frac{M^{12} \log (1/\delta)}{\epsilon^8}\right)$$
    Furthermore, because we evolve for time $t = \mathcal O(1/M)$, this means that each Pauli probe requires evolution time
    $$T = \mathcal O\left(\frac{M^{11} \log (1/\delta)}{\epsilon^8}\right)$$
    However, since our signal has potential support on $\mathcal O(M^2)$ possible Paulis, we require $K = \mathcal O(M^2 n \log^2 M)$ Pauli probes. Hence, the total evolution time is given as
    $$T = \mathcal O\left(\frac{M^{13}n \log^2 M}{\epsilon^8}\right)$$
\end{proof}

\chapter{Solving the Linear System}\label{chap:linear}

In the main text, we mentioned that the linear system for the coefficients obtained via the SPAM-robust technique can be solved efficiently. We will explicitly demonstrate this result in this chapter. From before, we had that each Pauli probe for a given Pauli $Q$ allowed us to noisily estimate the linear combination $-2t\sum_P \langle P, Q\rangle\gamma_P$.

\section{Compressed Sensing Problem}

This type of linear system we aim to solve falls under the umbrella of compressed sensing \cite{candes_stable_2006}. In particular, the compressed sensing problem as framed as follows:
\begin{problem}[Compressed Sensing]
    Suppose we are given noisy data $\mathbf y = \mathbf A\mathbf x_0 + \mathbf e$, where $\mathbf A \in \mathbb C^{K \times N}$, where $K \ll N$, $x_0 \in \mathbb C^N$ is taken to be $M$-sparse, meaning that there are at most $M$ nonzero coefficients in $\mathbf x_0$, and the error satisfies $\|\mathbf e\|_2 \le \eta$. How do we recover a $M$-sparse signal $x^\sharp$ such that
    $\|\mathbf x^\sharp - \mathbf x_0\|_2 \le \epsilon$, for some threshold $\epsilon$?
\end{problem}
Note that $\|\mathbf x^\sharp - \mathbf x_0\|_2 \le \epsilon \implies \|\mathbf x^\sharp - \mathbf x_0\|_\infty \le \epsilon$, meaning that the solution to this problem is sufficient for our purposes. This problem can be formulated as the following convex program:
$$\min \|\mathbf x\|_1 \quad \text{subject to} \quad \|\mathbf A\mathbf x - \mathbf y\|_2 \le \eta$$
Under certain assumptions on $\mathbf A$, it can be shown that solving this convex program yields a solution $\mathbf x^\sharp$ s.t.
$$\|\mathbf x^\sharp - \mathbf x_0\| \le C \eta$$
for some constant $C$ \cite{candes_stable_2006}. In particular, if $\mathbf A$ obeys the restricted isometry property (RIP), the bound above holds. This property is defined as follows:
\begin{definition}[Restricted Isometry Property]
    Let $\mathbf A \in \mathbb C^{K \times N}$ and $1 \le M \le K$ be an integer. If for every $M$-sparse vector $\mathbf v \in \mathbb C^N$ (i.e. $\|\mathbf v\|_0 \le M$), we have that
    $$(1 - \delta)\|\mathbf v\|_2^2 \le \|\mathbf A\mathbf v\|_2^2 \le (1 + \delta)\|\mathbf v\|_2^2,$$
    then $\mathbf A$ is said to be $(M, \delta)$-RIP.
\end{definition}
Consequently, if we can show that the measurement matrix used for the Pauli probes satisfies RIP with high probability for a reasonable number of measurements, we can determine a sufficient bound on error for the eigenvalue exponents to recover the Lindbladian coefficients to the desired accuracy.

\section{Relation with the Walsh-Hadamard Transform}

In order to demonstrate that our measurement matrix satisfies RIP, we consider its relation to the Walsh-Hadamard matrix.

\begin{definition}[Walsh-Hadamard Operator]
    Using the Hadamard operator $\mathbf H_1 = X + Z$, we define the Walsh-Hadamard operator $\mathbf H_n = \mathbf H_1^{\otimes n}$.
\end{definition}

We can then relate these Hadamard operators to the Pauli pattern matrix as follows.

\begin{proposition}
    The $n$-qubit Pauli Pattern matrix $\mathbf A_n = (a_{ij})$, where $a_{ij} = \langle P_i, P_j\rangle$ is related to the Walsh-Hadamard operator via
    $$\mathbf 1_n - 2\mathbf A_n = \mathbf H_{2n},$$
    for all $n \in \mathbb N$, where $\mathbf 1_n$ is a $4^n \times 4^n$ matrix of ones.
\end{proposition}

\begin{proof}
    We begin by observing that defining
    $$\mathbf B_n = \mathbf 1_n - 2\mathbf A_n,$$
    which gives us a matrix $\mathbf B_n$, where
    $$(\mathbf B_n)_{ij} = \begin{cases}
        +1 & [P_i, P_j] = 0 \\
        -1 & \{P_i, P_j\} = 0
    \end{cases}$$
    
    We then continue via induction. Consider the base case $n = 1$. Here, we have that
    $$\mathbf B_1 = \mathbf H_2 = \begin{pmatrix}
        1 & 1 & 1 & 1 \\
        1 & 1 & -1 & -1 \\
        1 & -1 & 1 & -1 \\
        1 & -1 & -1 & 1
    \end{pmatrix}$$
    Consequently, we trivially have that the base case holds. 

    For the induction step, we assume that this statement holds for $n = k - 1$. We can write any $k$-qubit Pauli in the form $P_i = P_i^{(-k)} \otimes P_i^{(k)}$, where $P_i^{(-k)}$ is an $n$-qubit Pauli. Let us define the bitstring corresponding to the $k - 1$-qubit Pauli as $i_{-k}$ and that for the Pauli on the last qubit as $i_k$. Consequently, we have that
    $$P_iP_j = (P_i^{(-k)} \otimes P_i^{(k)})(P_j^{(-k)} \otimes P_j^{(k)}) = (\mathbf B_{k - 1})_{i_{-k}j_{-k}}(\mathbf B_{1})_{i_{k}j_{k}} P_jP_i$$
    Since $(\mathbf B_{k})_{ij} = (\mathbf B_{k - 1})_{i_{-k}j_{-k}}(\mathbf B_{1})_{i_{k}j_{k}}$, we then have by our induction hypothesis that
    $$\mathbf B_k = \mathbf B_{k - 1} \otimes \mathbf B_1 = \mathbf H_{2(k - 1)} \otimes \mathbf H_2 = \mathbf H_{2k}$$
    Hence, since these pattern matrices have the same tensor product relation, we have that they are identical.
\end{proof}
Furthermore, recall that $\sum_P \gamma_P = 0$ by the definition of a Pauli Lindbladian. Consequently, we have that $\mathbf 1_{n}\boldsymbol \gamma = 0$. Consequently, if we scale each $\langle Q\rangle$ by $\frac{2}{t}$, we will obtain a linear system with measurement matrix given by $\mathbf H_{2n}$. Consequently, sampling random Pauli probes can be understood to effectively be sampling from rows of this Walsh-Hadamard transform.

We can demonstrate that a measurement matrix composed by subsampling from the Walsh-Hadamard transform satisfies RIP with high probability using the following theorem:
\begin{theorem}[Unitary Sub-Sampling \cite{haviv_restricted_2017}]\label{thm:sub_sample}
    For sufficiently large $N$ and $k$, a unitary matrix $U$ satisfying $\|U\|_\infty \le \mathcal O(1/\sqrt N)$, where $\| \cdot \|_\infty$ denotes the $\ell_\infty$ norm on the elements of $U$, and a sufficiently small $\delta > 0$, the following holds. For some $K = \mathcal O(\log^2 (1/\delta)\delta^{-2} \cdot M \cdot \log^2 (M/\delta) \cdot \log N)$, let $\mathbf A \in \mathbb C^{K \times N}$ be a matrix whose $K$ rows are chosen uniformly and independently from the rows of $U$, multiplied by $\sqrt{N/K}$. Then, with probability $1 - 2^{-\Omega(\log N \cdot \log(M/\delta))}$, the matrix $\mathbf A$ is $(M, \delta)$-RIP.
\end{theorem}
Since $\frac{1}{2^n}\mathbf H_{2n}$ is unitary and satisfies the $\|\cdot\|_\infty$ bound in Theorem \ref{thm:sub_sample}, we have that assuming a small constant $\delta \sim \mathcal O(1)$, we can construct a measurement matrix that is $(M, \delta)$-RIP with high probability using $K = \mathcal O(Mn\log^2 M)$ Pauli probes. Consequently, after rescaling the coefficient matrix, this means that if $\mathbf A \in \mathbb C^{K \times N}$ represents our subsampled Walsh-Hadamard Matrix and $\mathbf y \in \mathbb C^K$ is the vector of learned $\sum_{P} (1 - 2\langle P, Q\rangle) \gamma_P$ values,
$$\left\|\frac{1}{\sqrt K}(\mathbf A \boldsymbol \gamma - \mathbf y)\right\|_2 \le \eta \implies \|\hat {\boldsymbol \gamma} - \boldsymbol \gamma\|_2 \le C \eta,$$
where $\hat {\boldsymbol \gamma}$ is the learned signal. Consequently, choosing $\eta = \mathcal O(\epsilon)$ means that we have to learn $\mathbf y$ s.t.
$$\|\mathbf A \boldsymbol\gamma - \mathbf y \|_2 \le \mathcal O(\epsilon \sqrt K),$$
which can be accomplished by simply learning the result from each experiment to accuracy $\mathcal O(\epsilon)$. Hence, learning each exponent in the experiment to accuracy $\mathcal O(\epsilon)$ is sufficient to learn the coefficients to accuracy on the order of $\epsilon$.

\section{Compressed Sensing with Known Support}

Lastly, we analyze the compressed sensing setup when the support of the signal is known. Here, we assume that the signal has sparsity $M$, and the measurement matrix is given by $\mathbf A \in \mathbb R^{K \times N}$ and satisfies $(k, \delta)$-RIP, we have that the truncated matrix $\mathbf A' \in \mathbb R^{K \times M}$, where only the columns in the true support are kept, we have that for a given error vector $\mathbf e$ (where $\|\mathbf e\|_2 \le \eta$), we have that
$$\|(\mathbf A'^T\mathbf A')^{-1}\mathbf A'^T\mathbf e\|_2 \le \|(\mathbf A'^T\mathbf A')^{-1}\|_\text{op}\|\mathbf A'^T\mathbf e\|_2 \le \frac{\sqrt{1 + \delta}}{1 - \delta}\eta$$
Consequently, just as before, for $\eta = \mathcal O(\epsilon)$ sufficiently small and $\delta$ a small constant, we can estimate the solution to $\epsilon$-accuracy. Furthermore, since $\mathbf A' \in \mathbb R^{K \times M}$, the time complexity for this scales as $\mathcal O(KM^2) = \mathcal O(M^3n \log^2 M)$.

For our purposes, because we have to perform this for $\mathcal O(M^2)$ coefficients, this gives an overall post-processing time of $\mathcal O(M^5n \log^2 M)$.





\end{document}